\newif\ifpublic
\publictrue 

\documentclass[11pt]{article}

\usepackage{fullpage}
\usepackage{ttpalatino}
\usepackage{amsthm}
\usepackage{graphicx}
\usepackage{amssymb}
\usepackage{amsfonts}
\usepackage{amsmath}

\usepackage{pdfsync}
\usepackage{underscore}  
\theoremstyle{plain}
\newtheorem{lem}{Lemma}[section]
\newtheorem{theorem}[lem]{Theorem}
\newtheorem{maintheorem}[lem]{Main Theorem}
\newtheorem{lemma}[lem]{Lemma}

\newtheorem{corollary}[lem]{Corollary}
\newtheorem{conjecture}[lem]{Conjecture}
\newtheorem{claim}[lem]{Claim}
\newtheorem{proposition}[lem]{Proposition}
\newtheorem{definition}[lem]{Definition}

\theoremstyle{definition}
\newtheorem{remark}[lem]{Remark}

\ifpublic
\usepackage[disable]{todonotes}
\else
\usepackage[colorinlistoftodos]{todonotes}
\fi

\newcommand{\phnote}[1]{\todo[color=red!20!green!15, size=\footnotesize]{ph: #1}}
\newcommand{\inote}[1]{\todo[color=blue!25!, size=\footnotesize]{Irit: #1}}
\newcommand{\gnote}[1]{\todo[color=red!20!, size=\footnotesize]{Guy: #1}}

\newcommand{\mathprob}[1]{\mbox{\textmd{\textsc{#1}}}}

\newcommand{\poly}{\text{poly}}

\renewcommand{\epsilon}{\varepsilon}
\renewcommand{\phi}{\varphi}

\newcommand{\acc}{\mbox{\textmd{\textsc{acc}}}}
\newcommand{\rej}{\mathprob{rej}}
\newcommand{\etal}{{\it et~al.}}
\newcommand{\ACSAT}{\mathprob{Alg-CktSAT}}
\newcommand{\NP}{\mathprob{NP}}
\newcommand{\CSAT}{\mathprob{CktSAT}}
\newcommand{\SAT}{\mathprob{SAT}}

\newcommand{\LDE}{\mathprob{LDE}}
\newcommand{\Vi}[1]{\widetilde{#1}}
\newcommand{\Vix}{{\Vi{x}}}
\newcommand{\ViPi}{{\Vi\Pi}}
\newcommand{\QF}{{\rm QH}}

\newcommand\remove[1]{}
\newcommand\card[1]{\left| #1 \right|}
\newcommand\set[1]{\left\{ #1 \right\}}
\newcommand\parenth[1]{\left( #1 \right)}

\newcommand\sett[2]{\left\{ #1 \left| \; \vphantom{#1 #2} \right. #2  \right\}}

\newcommand{\func}{f}

\newcommand\ceil[1]{{\lceil #1 \rceil}}

\newcommand{\eps}{\varepsilon}
\renewcommand{\epsilon}{\varepsilon}

\newcommand{\ignore}[1]{}

\newcommand{ \bits    }{\{0,1\}}

\newcommand{ \D       }{\mathcal{D}}

\newcommand{\g}{{g}}

\newcommand{\integers}{\mathbb{Z}^+}
\renewcommand{\epsilon}{\varepsilon}
\newcommand{\zo}{\{0,1\}}

\newcommand{\agr}{\mathrm{agr}}
\newcommand{\spn}{\mathrm{span}}
\newcommand{\union}{\cup}

\newcommand{\outr}{{\mathrm{out}}}
\newcommand{\proj}{{\mathrm{proj}}}

\newcommand{\inn}{{\mathrm{in}}}
\newcommand{\comp}{\mathrm{comp}}

\newcommand{\ve }{{\hbox{   and   }}}
\newcommand{\F}{\mathbb{F}}

\newcommand{\e}{\epsilon}

\newcommand{\twocomp}{\circledast}

\newcommand{\din}{\delta_\inn}
\newcommand{\dout}{\delta_\outr}
\newcommand{\as}{{A_S}}
\newcommand{\ag}{{A_\gamma}}
\newcommand{\m}{{m_3}}
\newcommand{\fakePi}{{\widetilde \Pi}}
\newcommand{\fakex}{{\widetilde x}}

\usepackage[pagebackref,colorlinks=true,pdfpagemode=none,linkcolor=blue,citecolor=blue,filecolor=blue,
urlcolor=blue,colorlinks,breaklinks]{hyperref}  
\newcommand{\lref}[2][]{\hyperref[#2]{#1~\ref*{#2}}}
\renewcommand{\eqref}[1]{\hyperref[#1]{(\ref*{#1})}}
\numberwithin{equation}{section}

\begin{document}

\title{Polynomially Low Error PCPs with polyloglog n Queries via
  Modular Composition\thanks{A preliminary version of this paper
    appeared in the {\em Proc.\ $47$th ACM Symp.\ on Theory of
      Computing (STOC)}, 2015~\cite{DinurHK2015}.}}

\date{\today}
\author{
  Irit Dinur\thanks{Weizmann Institute of Science, ISRAEL. email: {\tt
      irit.dinur@weizmann.ac.il}. Research supported in part by a
    ISF-UGC grant 1399/4 and by an ERC grant 239985.}
  \and
  Prahladh Harsha\thanks{Tata Institute of Fundamental Research
    (TIFR), Mumbai, INDIA. email: {\tt prahladh@tifr.res.in}. Research
    supported in part by ISF-UGC grant 1399/4. Part of this work
was done while visiting the Simons Institute for the Theory of Computing, UC Berkeley.}
\and Guy Kindler\thanks{The Hebrew University of Jerusalem,
  ISRAEL. email: {\tt gkindler@cs.huji.ac.il}.  Research supported in part by an
    Israeli Science Foundation grant no. 1692/13 and by US-Israel Binational Science
    Foundation grant no. 2012220. Part of this work
was done while visiting the Simons Institute for the Theory of Computing, UC Berkeley.}}
\begin{titlepage}

\maketitle

\thispagestyle{empty}
\setcounter{page}{0}

\begin{abstract}
  We show that every language in NP has a PCP verifier that tosses
  $O(\log n)$ random coins, has perfect completeness, and a soundness error of at most
  $1/\poly(n)$, while making at most $O(\poly\log\log n)$ queries into a
  proof over an alphabet of size at most $n^{1/\poly\log\log n}$. Previous constructions that obtain $1/\poly(n)$ soundness error used
  either $\poly\log n $ queries or an exponential sized alphabet, i.e. of
  size $2^{n^c}$ for some $c>0$. Our result is an exponential
  improvement in both parameters simultaneously.

  Our result can be phrased as a polynomial-gap hardness for approximate
  CSPs with arity $\poly\log\log n$ and alphabet size $n^{1/\poly\log
    n}$. The ultimate goal, in this direction,
  would be to prove polynomial hardness for CSPs with constant arity and
  polynomial alphabet size (aka the sliding scale conjecture for inverse
  polynomial soundness error).

  Our construction is based on a modular generalization of previous
  PCP constructions in this parameter regime,
 which involves a
  composition theorem that uses an extra `consistency' query but
  maintains the inverse polynomial relation between the soundness
  error and the alphabet size.

Our main technical/conceptual contribution is a new notion of
soundness, which we refer to as {\em
    distributional soundness}, that replaces the previous notion of
  ``list decoding soundness'', and that allows us to prove a modular
  composition theorem with tighter parameters. This new notion of
  soundness allows us
  to invoke composition a super-constant number of
  times without incurring a blow-up in the soundness error.
\end{abstract}

\end{titlepage}

\section{Introduction}

Probabilistically checkable proofs (PCPs) provide a proof format that
enables verification with only a small number of queries into the
proof such that the verification, though probabilistic, has only a
small probability of error. This is formally captured by the
following notion of a probabilistic verifier.
\begin{definition}[PCP Verifier]\label{def:PCPver}
  A PCP verifier $V$ for a language $L$ is a polynomial time
  probabilistic algorithm that behaves as follows: On input $x$, and
  oracle access to a (proof) string $\Pi$ (over an alphabet $\Sigma$),
  the verifier reads the input $x$, tosses some random coins $R$, and
  based on $x$ and $R$ computes a (local) window $I=(i_1,\ldots,i_q)$
  of $q$ indices to read from $\Pi$, and a (local) predicate
  $\varphi:\Sigma^q \to \zo$. The verifier then accepts iff
  $\varphi(\Pi\vert_I)=1$.
\begin{itemize}
\item The verifier has {\em perfect completeness}: if for every $x\in
  L$, there is a proof $\Pi$ that is accepted with probability $1$. I.e., $ \exists \Pi ,\;\Pr_{R}[\varphi(\Pi\vert_I)=1]=1$.
\item The verifier has {\em soundness error} $\delta<1$: if for any $x\not\in L$, every proof $\Pi$ is
accepted with probability at most $\delta$. I.e., $\forall \Pi,\; \Pr_{R}[\varphi(\Pi\vert_I)=1]\le \delta$.
\end{itemize}
\end{definition}
The celebrated PCP Theorem~\cite{AroraS1998,AroraLMSS1998} states that
every language in NP has a verifier that has perfect completeness and
soundness error bounded by a constant $\delta< 1$, while using only a
logarithmic number of random coins, and reading only $q=O(1)$ proof
bits.  Naturally (and motivated by the fruitful connection to
inapproximability due to Feige~\etal~\cite{FeigeGLSS1996}), much attention has
been given to obtaining PCPs with desirable parameters, such as a
small number of queries $q$, smallest possible soundness error
$\delta$, and smallest possible alphabet size $\card \Sigma$.%

How small can we expect the soundness error $\delta$ to be? There are
a couple of obvious limitations. First observe that the soundness
error $\delta$ cannot be smaller than $1/\poly(n)$ just because there
are only $\poly(n)$ different random choices for the
verifier and at least one of the corresponding local predicates must
be satisfiable\footnote{One may assume that every
  local predicate $\phi$ is
  satisfiable. Otherwise the question of ``$x\stackrel{?}{\in} L$''
  reduces to the question of whether $\varphi$ is satisfiable for {\em
    any} of the predicates computed by the verifier. This cannot occur
  without a collapse of NP into
  $\mathprob{NTIME}(q\log\card\Sigma)$.}. Next, note that if the
verifier reads a total of ${k}$ bits from the proof (namely,
$q\log\card\Sigma \le {k}$), the soundness error cannot be smaller
than $2^{-k}$, just because a random proof is expected to cause the
verifier to accept with at least this probability.

The best case scenario is thus if one can have the verifier read ${k}
= O(\log n)$ bits from the proof and achieve a soundness error of
$1/2^k = 1/\poly(n)$. Indeed, the following is well known (obtained by
applying a randomness efficient sequential repetition to the basic PCP Theorem):
\begin{theorem}[PCP theorem + randomness efficient sequential repetition]\label{thm:basic-seqrep}
For every integer $k$, every language in NP has a PCP verifier that
tosses at most $O(k+\log n)$ random coins, makes $q=O({k})$ queries into a
proof over the Boolean alphabet $\bits$, has perfect completeness, and
soundness error $\delta = 2^{-{k}}$.

In particular, setting ${k} = \log n$ we get $q = O(\log n)$ and $\delta = 1/\poly(n)$.
\end{theorem}

This theorem gives a ballpark optimal tradeoff (up to constants)
between soundness error and the number of bits read from the
proof. However it does not achieve a {\em small number of queries}, a
fundamental requirement that is important, among other things, for
hardness of approximation.
The goal of constructing a PCP with both a small error and a {small
  number of queries} turns out to be much more challenging and has
attracted considerable attention. This was first formulated by
Bellare~\etal~\cite{BellareGLR1993} as the ``sliding
scale'' conjecture.
\begin{conjecture}[Sliding Scale
  Conjecture~\cite{BellareGLR1993}]
For any $\frac 1{\poly(n)}\le \delta <1$, every language in NP  has a
PCP verifier that tosses $O(\log n)$ random coins, makes
$q=O(1)$ queries\footnote{It is even conjectured that this constant can be made
  as low as 2.} into a proof over an alphabet $\Sigma$ of size
$\poly(1/\delta)$, has perfect completeness, and soundness error $\delta$.
\end{conjecture}

As we describe shortly below, this conjecture is known to hold for $1>
\delta \ge 2^{-(\log n)^{1-\eps}}$, namely where $\delta$ can be made
`almost' polynomially small. The interesting regime, that has remained
open for two decades, is that of (inverse) polynomially small
$\delta$. This is the focus of our work. Our main goal is to find the
smallest $q$ and $\card\Sigma$ parameters for which we can get
$\delta$ to be polynomially small. Our main result is the following.
\begin{maintheorem}\label{thm:main}
  Every language in NP has a PCP verifier that tosses $O(\log n)$
  random bits, makes $q= (\log\log n)^{O(1)}$ queries into a proof
  over an alphabet $\Sigma$ of size $\card\Sigma =n^{1/ {(\log\log
      n)}^{O(1)}}$, has perfect completeness, and soundness error
  $\delta=1/\poly(n)$.
\end{maintheorem}
Previous PCP constructions require at least $(\log n)^{\Omega(1)}$
queries in order to achieve polynomially small error (and this remains
true even for constructions that are allowed quasi-polynomial size,
see further discussion at the end of this introduction).

The first works making progress towards this conjecture are due to
Raz and Safra~\cite{RazS1997}, and Arora and Sudan~\cite{AroraS2003}, and rely on the classical
(algebraic) constructions of PCPs. They prove the conjecture for all
$\delta$ such that $\delta \ge 2^{-(\log n)^\beta}$ for some constant
$\beta>0$. These ideas were then extended by Dinur~\etal~\cite{DinurFKRS2011} with an
elaborate composition-recursion structure, proving the conjecture for
all $\delta \ge 2^{-(\log n)^{1-\eps}}$ for {\em any} $\eps>0$. The
small catch here is that the number of queries grows as $\eps$
approaches $0$. The exact dependence of $q$ on $\eps$ was not
explicitly analyzed in~\cite{DinurFKRS2011}, but we show that it can
be made $O(1/\eps)$ while re-deriving their result.
\begin{theorem}[\cite{DinurFKRS2011}]
  For every $\eps>0$ and $\delta = 2^{-(\log n)^{1-\eps}}$, every language in NP has a PCP verifier that
  tosses $O(\log n)$ random coins, makes $q= O(1/\eps)$ queries into
  a proof over an alphabet $\Sigma$ of size $\card\Sigma=
  1/\poly(\delta)$, has perfect completeness, and has soundness error
  $\delta$.
\end{theorem}

The focus of \cite{DinurFKRS2011} was on a constant number
of queries but their result can also be applied towards getting
polynomially small error with a non-trivially small number of
queries. This is done by combining it with sequential repetition. We
get,
\begin{corollary}[\cite{DinurFKRS2011} + randomness efficient
  sequential repetition]\label{cor:DFKRS} For every $\eps>0$, every
  language in NP has a PCP verifier that
  tosses $O(\log n)$ random coins, makes $q= O((\log n)^\eps/\eps)$ queries into
  a proof over an alphabet $\Sigma$ of size $\card\Sigma= 2^{(\log n)^{1-\eps}}$, has perfect completeness, and has soundness error
  $\delta=1/\poly(n)$.
\end{corollary}
\lref[Corollary]{cor:DFKRS} describes the previously known best result
in terms of minimizing the number of queries subject to achieving a
polynomially small error and using at most a logarithmic amount of
randomness. Whereas in \lref[Corollary]{cor:DFKRS} the number of
queries is $q=(\log n)^\eps$, our \lref[Main Theorem]{thm:main} requires
only $q=\poly\log\log n$ queries.
\subsection*{PCP Composition and dPCPs}
Like in recent improved constructions of PCPs
\cite{BenSassonGHSV2006,DinurR2006,BenSassonS2008,Dinur2007,MoshkovitzR2010b,DinurH2013}, our main theorem is obtained
via a better understanding of composition.
All known constructions of PCPs rely on proof composition. This
paradigm, introduced by Arora and Safra~\cite{AroraS1998}, is a
recursive procedure applied to PCP constructions to reduce the
alphabet size. The idea is to start with an easier task of
constructing a PCP over a very large alphabet $\Sigma$. Then,
proof composition is applied (possibly several times over) to PCPs
over the large alphabet to obtain PCPs over a smaller (even binary)
alphabet, while keeping the soundness error small.

In the regime of high soundness error (greater than 1/2), composition
is by now well understood using the notion of {\em PCPs of
  proximity}~\cite{BenSassonGHSV2006} (called {\em assignment testers}
in \cite{DinurR2006}) (see also \cite{Szegedy1999}). The idea is to
bind the PCP proof of a statement to an NP witness for it, so that the
verifier not only checks that the statement is correct but also that
the given witness is (close to) a valid one. This extension allows one
to prove a modular composition theorem, which is oblivious to the
inner makings of the PCPs being composed. This modular approach has
facilitated alternate proofs of the PCP theorem and constructions of
shorter PCPs~\cite{BenSassonGHSV2006,BenSassonS2008,Dinur2007}. However,
the notion of a PCP of proximity, or assignment tester, is not useful
for PCPs with low-soundness error. The reason is that for small
$\delta$ we are in a ``list decoding regime'', in that the PCP proof
can be simultaneously correlated with more than one valid NP witness.

The works mentioned earlier \cite{RazS1997,AroraS2003,DinurFKRS2011}
addressed this issue by using a notion of local list-decoding. This
was called a local-reader in~\cite{DinurFKRS2011} and
formalized nicely as a locally-decode-or-reject-code (LDRC) by
Moshkovitz and Raz~\cite{MoshkovitzR2010b}. Such a code allows ``local
decoding'' in that for any given string $w$ there is a list of valid
codewords $\set{c_1,\ldots,c_L}$ such that when the verifier is given
a tuple of indices $j_1,\ldots,j_k$, then but for an error probability
of $\delta$, the verifier either rejects or outputs
$(c_i)|_{j_1,\ldots,j_k}$ for some $i\in [L]$.

\paragraph{Decodable PCPs (dPCPs)} Dinur and Harsha~\cite{DinurH2013}
introduced the notion of a PCP decoder (dPCP), which extends the
earlier definitions of LDRCs and local-readers from codes to PCP
verifiers.  A PCP decoder is like a PCP verifier except that it also
gets as input an index $j$ (or a tuple of indices). The PCP decoder is
supposed to check that the proof is correct, and to also return the
$j$-th symbol of the NP witness encoded by the PCP proof.  As in
previous work, the soundness guarantee for dPCPs is that for any given
proof $\Pi$ there is a {\em short} list of valid NP witnesses
$\set{x_1,\ldots,x_L}$ such that except with probability $\delta$ the
verifier either rejects or outputs $(x_i)|_j$ for some $x_i$ in the
list.

The main advantage of dPCPs is that they allow a modular composition
theorem in the regime of small soundness error. The composition
theorem proved by Dinur and Harsha~\cite{DinurH2013} was a two-query
composition theorem, generalizing from the ingenious construction of
Moshkovitz and Raz~\cite{MoshkovitzR2010b}. The two-query requirement
is a stringent one, and in this construction it inherently causes an
exponential increase in the alphabet, so that instead of $\card\Sigma
= \poly(1/\delta)$ one gets a PCP with $\card\Sigma =
\exp(\poly(1/\delta))$.

In this work we give a different (modular) dPCP composition
theorem. Essentially, our theorem is a modular generalization of the
composition method, as done implicitly in previous
works~\cite{RazS1997,AroraS2003,DinurFKRS2011}, which uses an extra
`consistency' query but maintains the inverse polynomial relation
between $\delta$ and $\card\Sigma$.

We remark that unlike recent PCP constructions
\cite{BenSassonGHSV2006,DinurR2006,MoshkovitzR2010b,DinurH2013} which
recurse on the large (projection) query of the outer PCP, our
composition recurses on the entire test as was done originally by
Arora and Safra~\cite{AroraS1998}. This aspect of the
composition is explained  and abstracted nicely in \cite{Moshkovitz2014}.

\paragraph{Distributional Soundness}
Our main technical/conceptual contribution is a new notion of
soundness, which we refer to as {\em distributional soundness}, which
replaces the previous notion of {\em list decoding soundness} described above, and
allows us to apply a non-constant number of compositions without a
blowup in the error.

We say that a verifier has {\em distributional soundness $\delta$} if
its output is ``$\delta$-indistinguishable'' from the output of an
idealized verifier. The idealized verifier has access to a
distribution $\widetilde\Pi$ over {\em valid proofs} or $\bot$. When
it is run with random coins $R$ it samples a proof $\widetilde\Pi(R)$
from this distribution and either rejects if $\widetilde\Pi(R)=\bot$
or outputs what the actual verifier would output when given access to
$\widetilde\Pi(R)$. By $\delta$-indistinguishable, we mean that there
is a coupling between the actual verifier and the idealized verifier,
such that the probability that the verifier does not reject {\em and}
its output differs from the output of the idealized verifier, is at
most $\delta$.

The advantage of moving from list decoding soundness to distributional
soundness, is that it removes the extra factor of $L_\inn$ (the list
size) incurred in previous composition analyses. Recall that, e.g. in
the composition theorem of Dinur-Harsha~\cite{DinurH2013}, one takes an outer PCP
with soundness $\dout$ and an inner PCP decoder with soundness $\din$
and out comes a PCP with soundness $\din+L_\inn \cdot \dout$. This is
true in all (including implicit) prior composition analyses. When
making only a constant number of composition steps, this is not an
issue, but when the number $t$ of composition steps grows, the
soundness is at least $(L_\inn)^t \cdot\dout$ and this is too
expensive for the parameters we seek.  Using distributional soundness,
we prove that the composition of a PCP with soundness error $\dout$
and a dPCP with soundness error $\din$ yields a PCP with soundness
error $\dout+\din + \eta$ where $\eta$ is an error term that is
related to the distance of an underlying error correcting code, and
can be controlled easily. Thus, after $t$ composition steps the
soundness error will only be $O(t(\delta+\eta))$.

We remark that this notion of soundness, though new, is satisfied by
most of the earlier PCP constructions (at least the ones based on the
low-degree test). In particular, it can be shown that list-decoding
soundness and very good error-correcting properties of the PCP imply
distributional soundness.

\subsection*{Proof Overview}
At a high level, our main theorem is derived by adopting the recursive
structure of the construction in~\cite{DinurFKRS2011}. The
two main differences are the use of our modular composition theorem,
and the soundness analysis that relies on the notion of distributional
soundness\footnote{Looking at the construction as it is presented
  here, one may ask, why wasn't this  done back in 1999, when the
  conference version of~\cite{DinurFKRS2011} was
  published.  This construction is notoriously far from
  modular. Thus, tweaking parameters, following them throughout the
  construction, and making the necessary changes, would have been a
  daunting task. Without the modular approach it was not at all clear
  what the bottlenecks were, let alone address them.  }.

We fix a field $\F$ at the outset and use the same field throughout the construction. 
This is important for the interface between the outer and the inner
dPCPs, as it provides a convenient representation of the output of the
outer dPCP as an arithmetic circuit over $\F$, which is then the input
for the inner dPCP.

As in the construction of~\cite{DinurFKRS2011}, we take $\card\F\approx  2^{(\log n)^{1-\eps}}$
and begin by constructing PCPs over a fairly large alphabet size which
we gradually reduce via composition. The initial alphabet size is
$2^{2^{(\log n)^{1-\eps}}}$, and then it drops to $2^{2^{(\log
  n)^{1-{\bf 2}\eps}}}$ and then to $2^{2^{(\log n)^{1-{\bf
      3}\eps}}}$ and so on. After $1/\eps$ steps we make a couple of
final composition steps and end up with the desired alphabet size of
$2^{(\log n)^{1-\eps}}$, logarithmic in the initial alphabet size.

Unlike the construction in~\cite{DinurFKRS2011}, we can afford to plug
in a sub-constant value for $\eps$, and we take $\eps = c\log\log \log
n / \log\log n$ for some constant $c$ so that $2^{(\log
  n)^{1-\eps}}=2^{\log n /(\log\log n)^c} = n^{1/(\log\log n)^c}$.

The number of composition steps is $O( 1/\eps )$, resulting in a PCP
with $O(\log\log n)$ queries and soundness error $n^{1/(\log\log n)^c}$
for some constant $c<3$ (see \lref[Theorem]{thm:maincons}). Finally,
$(\log\log n)^c$ steps of (randomness-efficient) sequential repetition
yield a PCP with polynomially small error and $\poly\log\log n$
queries as stated in \lref[Main Theorem]{thm:main}. It can be shown
that the parameters obtained in \lref[Theorem]{thm:maincons} (namely,
soundness error $n^{1/(\log \log n)^{O(1)}}$) is tight given the basic Reed-Muller and Hadamard based building blocks (see
\lref[\S]{sec:optimal} for details).


\subsection*{Further Background and Motivation}

Every PCP theorem can be viewed as a statement about the local-vs.-global behaviour of proofs, in that the correctness of a proof, which is clearly a {\em global} property, can be checked by {\em local} checks, on average. The parameters of the PCP (number of queries, soundness error, alphabet size) give a quantitative measure to this local to global behavior. The sliding scale conjecture essentially says that even with a {\em constant} number of queries, this local to global phenomenon continues to hold, for all ranges of the soundness error.

Another motivation for minimizing the number of queries becomes apparent
when considering interaction with provers instead of direct access to
proofs (i.e. MIP instead of PCP). A PCP protocol can not in general be
simulated by a protocol between a verifier and a prover because the
prover might cheat by adaptively changing her answers. This can be sidestepped
by sending each query to a different prover, such that the provers are
not allowed to communicate with each other. This is the MIP model of
Ben-Or~\etal~\cite{BenOrGKW1988}. It is only natural to seek protocols using the smallest number of (non-communicating)
provers.

The importance of the sliding scale conjecture stems, in addition to
the fundamental nature of the question, from its applications to
hardness of approximation. First, it is known that every PCP theorem
can be phrased as a hardness-of-approximation for {\sc Max-CSP}: the
problem of finding an assignment that satisfies a maximal number of
constraints in a given constraint system. The soundness error
translates to the approximation factor, the alphabet of the proof is
the alphabet of the variables, and the number of queries becomes the
arity of the constraints in the CSP.

The main goal of this paper can be phrased as proving polynomial
hardness of approximation factors for CSPs with smallest possible
arity (and over an appropriately small alphabet). Our main theorem translates to the following result
\begin{theorem}\label{thm:main}
It is NP-hard to decide if a given CSP with $n$ variables and $\poly(n)$ constraints, is perfectly satisfiable or whether every assignment satisfies at most $1/\poly(n)$ fraction of the constraints. The CSP is such that each constraint has arity at most $ \poly\log\log n$ and the variables take values over an alphabet of size at most $n^{1/\log\log n^{O(1)}}$. 
\end{theorem}

In addition to the syntactic connection to {\sc Max-CSP}, it is also
known that a proof of the sliding scale conjecture
would immediately imply polynomial factors inapproximability for {\sc
  Directed-Sparsest-Cut} and {\sc Directed-Multi-Cut}~\cite{ChuzhoyK2009}.

The results of \cite{DinurFKRS2011} were used in \cite{DinurS2004} for
proving hardness of approximation for a certain $\ell_p$ variant of
label cover. However, that work mis-quoted the main result from
\cite{DinurFKRS2011} as holding true even for a super-constant number
of queries, up to $\sqrt{\log \log n}$. In this work, we fill the gap
proving the required PCP statement. We thank Michael Elkin for
pointing this out. 

\subsection*{Further Discussion}
A possible alternate route to small soundness PCPs is via the
combination of the basic PCP theorem~\cite{AroraS1998,AroraLMSS1998}
with the parallel repetition theorem~\cite{Raz1998}. Applying $k$-fold
parallel repetition yields a two-query PCP verifier over alphabet of
size $\card\Sigma=2^{O(k)}$, that uses $O(k\log n)$ random bits, and
has soundness error $\delta= 2^{-\Omega(k)}$.

If we restrict to polynomial-size constructions, then parallel
repetition is of no help compared to \lref[Theorem]{thm:basic-seqrep}.
If we allow $k$ to be super constant, then more can be
obtained. First, it is important to realize that the soundness error
should be measured in terms of the output size, namely $N = 2^{k\log
  n}$. For $k = (\log n)^c$ a simple calculation shows $\log n = (\log
N)^{1/(c+1)}$, and hence the soundness error is $\delta(N) =
2^{-(\log n)^c}=2^{-(\log N)^{1-1/(c+1)}}$. This is no better than the
result of \cite{DinurFKRS2011} in terms of the soundness error, and in
fact, worse in terms of the instance size blow up ($N=2^{(\log
  n)^{c+1}}$ as opposed to $N= n^{O(1)}$). Even parameters similar to
our main theorem can be obtained, albeit with an almost exponential
blowup. Consider $k = \sqrt n$ for example. In this case $\log n = 2
\log \log N - \Theta(\log \log \log N)$, and so $\delta(N) =
2^{-\sqrt{n}}=N^{-1/\Theta(\log\log N)}$. From here, to get a
polynomially small error one can take $O(\log\log N)$ rounds of
(randomness efficient) sequential repetition, coming up with a result that is similar to our
\lref[Main Theorem]{thm:main} but with a huge blow up $(N=
n^{\sqrt{n}}$ as opposed to $N=n^{O(1)}$).

We remark that a natural approach towards the sliding scale conjecture
is to try and find a {\em randomness-efficient} version of parallel
repetition to match the parameters of \lref[Theorem]{thm:basic-seqrep}
but with $q=O(1)$. Unfortunately, this approach has serious
limitations~\cite{FeigeK1995} and has so-far been less successful than
the algebra-and-composition route, see also \cite{DinurM2011,
  Moshkovitz2014}.

\subsection*{Organization}
We begin with some preliminaries in \lref[\S]{sec:prelim}. We
introduce and define dPCPs and distributional soundness in
\lref[\S]{sec:defs}. Our dPCPs have $k$ provers which are analogous to
(and stronger than) PCPs that make $k$ queries. In
\lref[\S]{sec:composition}, we state and prove a modular composition
theorem for two (algebraic) dPCPs. In \lref[\S]{sec:together},
we prove the main theorem, relying on specific ``classical''
constructions of PCPs that are given in \lref[\S]{sec:components}, one
based on the Reed-Muller code and low degree test, and one based on the
quadratic version of the Hadamard code. These PCPs are the same as in
earlier
constructs~\cite{RazS1997,AroraS2003,DinurFKRS2011,MoshkovitzR2010b,DinurH2013}
except that here we prove that they have the stronger notion of small
{distributional soundness}.

\section{Preliminaries}\label{sec:prelim}

\subsection{Notation}

All {\em circuits} in this paper have fan-in 2 and fan-out 2, and we
allow only unary NOT and binary AND Boolean operations as internal
gates. The {\em size} of a Boolean circuit/predicate $\Phi$ is the
number of gates in $\Phi$.  Given a circuit/predicate
$\Phi:\set{0,1}^n\to \set{0,1}$, we denote by $\SAT(\Phi)$ the set of
satisfying assignments for $\Phi$, i.e.,
$$\SAT(\Phi) = \sett{x\in \set{0,1}^n}{\Phi(x) = 1 }.$$  We will refer to
the following $\NP$-complete language associated with circuits:
$$\CSAT = \sett{ \Phi}{\Phi \text{ is specified as a Boolean circuit
    and } \SAT(\Phi) \neq \emptyset}.$$
We will follow the following convention regarding input lengths: $n$
will refer to the length of the input to the circuit $\Phi$ (i.e.,
$\Phi:\set{0,1}^n\to\set{0,1}$) while $N$ will refer to the size of
the circuit/predicate (i.e., $\text{size}(\Phi) = N$). Thus, $N$ is
the input size to the problem $\CSAT$.

We will also refer to a similar language associated with arithmetic
circuits.  First, for some notation. Given a finite field $\F$, we
consider arithmetic circuits over $\F$ with addition ($+$) and
multiplication ($\times$) gates and constants from the field $\F$. For
a function $\Phi:\F^n \to \F$, the
size of $\Phi$ is the number of gates in the arithmetic
circuit specifying $\Phi$. We denote by $\SAT(\Phi)$ the set of all
$x$ such that $\Phi(x) =0$.

\begin{definition}[Algebraic Circuit SAT]\label{def:acsat} Given a field $\F$, the
  {\em Algebraic-Circuit-Satisfiability} problem, denoted by $\ACSAT_\F$,
  is defined as follows:
\[\ACSAT_\F = \sett{ \Phi}{\Phi \text{ is specified
    by an arithmetic circuit over } \F \text{ and } \SAT(\Phi) \neq \emptyset}.\]
As in the case of $\CSAT$, $n$ refers to the length of the input to
the function $\Phi$ (i.e., $\Phi: \F^n \to \F$), while $N$ refers to
the size of the arithmetic circuit $\Phi$.
\end{definition}

\subsection{Error Correcting Codes}
Let $E:\F^n\to \F^N$ be an error correcting code with relative
distance $1-\mu$, i.e., for every $x\neq x'$, $\Pr_{j\in
  [N]}[E(x)_j = E(x')_j ] \le \mu$. For a word $w\in \F^N$ that is not
necessarily a correct codeword, we can consider the list of all ``admissible'' codewords, i.e. codewords that have a non-negligible correlation with $w$.
We are interested in more than just a list: we want to associate with each index $j\in [N]$ an element in that list in a unique way. This will allow us to treat $w$ as a random variable $W$: for a random index $j$, the random variable $W(j)$ will output the list-element associated with the $j$th index.
\begin{definition}\label{def:agreedist}
Let $\tau>0$ be a parameter an let $w:[N]\to \F$. We define the {\em
  $\tau$-local
decoding function of $w$ with respect to the code $E$},\   $W:[N]\to \F^n\cup \set\bot$, as follows:
\begin{itemize}
\item The $\tau$-admissible words for $w$ are  $$agr_\tau(w) = \sett{ x\in \F^n}{\Pr_{j\in [N]}[E(x)_j = w_j ] \ge \tau }.$$
\item For any $j\in[N]$, if there is a unique word $x\in
  \agr_\tau(w)$ such that $E(x)_j = w_j$ we set $W(j) =
  x$. Otherwise, we set $W(j)=\bot$.
\end{itemize}
\end{definition}

\begin{claim}\label{claim:ambiguous}
Let $E:\F^n\to \F^N$ be an error correcting code with relative
distance $1-\mu$, let $w:[N]\to F$, and let  $W:[N]\to \F^n$ be its
$\tau$-local decoding function with respect to $E$. Also, suppose
that $v:[N]\to \F$ is a legal codeword, i.e., $v=E(y)$ for some $y \in
\F^n$. Then
\[
\Pr_{j\in [N]}\left[ v_j=w_j \ve W(j)\neq y  \right] \leq \tau+4\mu/\tau^2.
\]
\end{claim}
\begin{proof}
  Without loss of generality, we may assume that $\tau \ge 2\sqrt \mu
  $. Let us write
  $\agr_\tau(w) =\set{x_1,x_2,\ldots }$ and let $S_i = \sett{j\in[N]}{w_j = (E(x_i))_j}$. We say that $j\in[N]$ is
  an ambiguous point for $w$ if $j\in S_i \cap S_{i'}$ for some
  distinct $i,i'$. We first bound the fraction of ambiguous points for
  $w$.

  By inclusion-exclusion,
  for any $\ell \le \card{\agr_\tau(w)}$
  \[ N \ge \card {\cup_{i=1}^\ell S_i} \ge \sum_{i=1}^\ell \card{S_i}
  - \sum_{i\neq {i'}\le \ell} \card{S_{i}\cap S_{i'}} ) \ge \ell \tau N -
  \binom{\ell }2 \mu N, \] where we have used that by definition
  $\card {S_i} \ge \tau N$ and by the distance of the code $\card {S_i
    \cap S_{i'}}\le \mu N$. This implies that for every
  $\ell\le\card{\agr_\tau(w)}$,
\[
1 \ge \ell \tau - \binom{\ell}2 \mu \ge  \ell \tau - \ell^2 \tau^2 /8
\]
which clearly fails if $\ell = 2/\tau $, so $\card{\agr_\tau(w)}\le 2/\tau$ and so
\begin{equation}
  \label{eq:1}
\Pr_{j\in [N]}[j\text{ is ambiguous}]\le  \frac 1 N\sum_{i\neq
  {i'}}\card{S_i\cap S_{i'}} \le 4\mu /\tau^2
\end{equation}

Now the event $\left[ v_j=w_j \ve W(j)\neq y  \right]$ can occur
either if $j$ is an ambiguous point for $w$, or if $v$ is not
$\tau$-admissible with respect to $w$. But the former happens with
probability at most $4\mu /\tau^2$ by \eqref{eq:1}, and the latter
happens with probability at most $\tau$, as otherwise $v$ would have
been $\tau$-admissible.
\end{proof}

\begin{remark}\label{rem:agreement}
We minimize the quantity $\tau + 4\mu/\tau^2$, by setting $\tau =
(4\mu)^{1/3}$. We refer to this minimum as the agreement parameter $\eta$
of the
code $E$. Thus, $\eta = 2\tau = 2(4\mu)^{1/3}$.
\end{remark}

\section{PCPs with distributional soundness}\label{sec:defs}

\subsection{Standard PCPs}

We begin by recalling the definition of a standard $k$-prover
projection PCP verifier.

\begin{definition}[PCP verifier]\label{def:pcpverifier}\
\begin{itemize}
\item A {\em $k$-prover projection PCP verifier} over alphabet $\F$
is a probabilistic-time algorithm $V$
that on input $\Phi$, a circuit of size $N$ and a random input $R$  of
$r(N)$ random bits generates a tuple $(q, \phi, g)$ where $q = (u,v_1\dots,v_{k-1})$ is a vector of $k$
queries, $\phi: \F^{m}\to \set{0,1}$ is a predicate, and $g =
(g_1,\dots,g_{k-1})$ is a list of $k-1$ functions $g_i : \F^{m} \to \F$ such that the size of the tuple $(\phi, g)$ is at most $s(N)$.
\item We write $(q,\phi,g) = V(\Phi; R)$ to denote the
  query-predicate-function tuple output by the verifier  $V$ on input
  $\Phi$ and random input $R$.
\item It is good to keep in mind the $k=2$ case as it captures all of the difficulty. In this case the output of $V$ is a {\em label cover} instance, when enumerating over all of $V$'s random inputs. (The query pairs specify edges $(u,v)$ and $(\varphi,g)$ specify which pairs of labels are acceptable).
\item We think of $V$ as a probabilistic oracle machine that on input
  $(\Phi;R)$ queries $k$ provers $\Pi=(A, B_1, \dots, B_{k-1})$ at positions
  $q = (u,v_1,\dots,v_{k-1})$ respectively to receive the answers
  $\Pi|_q := $ $(A(u),B_1(v_1),\dots, B_{k-1}(v_{k-1}))$ $\in \F^m \times
  \F^{k-1}$, and accepts iff the following checks pass: $\phi(A(u))
  =1$ and $g_i(A(u)) = B_i(v_i)$ for all $i \in \{1,\dots, k-1\}$.
\item Given $k$ provers $\Pi=(A,B_1,\dots, B_{k-1})$, we will sometimes
  collectively refer to them as the ``proof $\Pi$'' . Furthermore, we refer to $A$ as the large
  prover and the $B_i$'s as the projection provers. We call $\Pi|_q$
  the local view of the proof $\Pi$ on queries $q$ and denote by
  $V^{\Pi}(\Phi; R)$ the output of the verifier $V$ on input $(\Phi; R)$
  when interacting with the $k$ provers $\Pi$. Thus, $V^{\Pi}(\Phi;R) =
  \acc$ if the checks pass and is $\rej$ otherwise.
\item We call $N$ the {\em input size}, $k$ the {\em number of
    provers}, $r(N)$ the {\em randomness complexity}, and $s(N)$ the
  {\em answer size} of the verifier $V$.
\end{itemize}
\end{definition}

\begin{definition}[standard PCPs]\label{def:pcpsound} For a function $\delta: \integers
  \to [0,1]$, a $k$-prover projection PCP verifier $V$ is a {\em
    $k$-prover probabilistically checkable proof system} for $\CSAT$
  with soundness error $\delta$ if the following completeness and
  soundness properties hold for every circuit $\Phi$:
\begin{description}
\item[Completeness:] If $x \in \SAT(\Phi)$, then there exist $k$ provers $\Pi =
  (A,B_1,\dots, B_{k-1})$ that cause the verifier $V$ to accept with
  probability 1. Formally, $$\exists \Pi = (A, B_1,\dots, B_{k-1}),
  \qquad\Pr_{R}\left[V^{\Pi}(\Phi;R) = \acc \right] =1.$$ In this
  case, we say that $\Pi$ is a {\em valid proof} for the statement $x\in
\SAT(\Phi)$.
\item[Soundness:] If $\Phi \notin \CSAT$ (i.e, $\SAT(\Phi) =
  \emptyset$), then for every $k$ provers $\Pi =
  (A,B_1,\dots,B_{k-1})$, the verifier $V$ accepts $\Phi$ with
  probability at most $\delta(N)$. Formally, $$\forall \Pi =
  (A,B_1,\dots,B_{k-1}),\qquad \Pr_{R}\left[ V^\Pi(\Phi;R) = \acc
  \right] \leq \delta(N).$$
\end{description}
We then say that $\CSAT$ has a $k$-prover projective PCP with soundness
error $\delta$.
\end{definition}

\subsection{Distributional Soundness}

We now present {\em distributional soundness}, a strengthening of the standard PCP soundness
condition that we find to be very natural. Informally, distributional soundness means that the event of the verifier accepting is roughly the same as the event of the local view of the verifier being consistent with a globally consistent proof, up to "the soundness error". I.e.,
\[ \Pr [\hbox{ accept }] = \Pr[\hbox{ accept ~~and~~ the local view agrees with a correct proof }]\pm \delta .
\]
Thus, the local acceptance of the verifier is "fully explained" in terms of global consistency.

A little more formally, every purported proof $\Pi$ (valid or not) can be coupled with an ``idealized'' distribution $\ViPi(R)$ over valid
proofs and $\bot$ such that the behavior of the verifier on random
string $R$ when interacting with the  proof $\Pi$ is  identical
to the corresponding behavior when interacting with the ``idealized'' proof
$\ViPi(R)$ upto an error of $\delta$, which we call the distributional
soundness error. Formally,

\begin{definition}[Distributional Soundness for $k$-prover PCPs]\label{def:pcpdistsound} For a function $\delta: \integers
  \to [0,1]$, a $k$-prover projection PCP verifier $V$ for $\CSAT$
  is said to have {\em distributional soundness error} $\delta$ if for
  every circuit $\Phi$ and any set of provers $\Pi =
  (A,B_1,\dots,B_{k-1})$ there is an `idealized pair' of functions $\Vix(R)$ and
  $\ViPi(R)$ defined for every random string $R$ such that the following holds.
\begin{itemize}
\item For every random string $R$,
  $\ViPi(R)$ is either a valid proof for the statement $\Vix(R)\in
  SAT(\Phi)$, or $\ViPi(R)=\bot$.
\item With probability at least
  $1-\delta$ over the choice of the random string $R$, the local view
  of the provers $\Pi$ completely agrees with the local view of the
  provers $\ViPi(R)$ or is a rejecting local view. In other words,
$$\Pr_R \left[ V^\Pi(\Phi; R) = \bot
  \text{ or } \Pi|_q = \ViPi(R)|_q \right] \geq 1-\delta,$$
where $q$ is the query vector generated by the PCP verifier $V$ on
input $(\Phi;R)$.

When the local views agree, i.e. $\Pi|_q = \ViPi(R)|_q $, we say that $\ViPi$ is {\em successful} (in explaining the success of $\Pi$).
\end{itemize}
\end{definition}
The advantage of distributional soundness is that it explains the
acceptance probability of every proof $\Pi$, valid or otherwise, in
the following sense. Suppose a proof $\Pi$ is accepted with probability
$p$. I.e., $p$ fraction of the local views $\Pi|_q$ are
``accepting''. Then, it must be the case that but for an error
probability of $\delta$, each of these accepting views are projections
of (possibly different) valid proofs. It is an easy consequence of
this, that distributional soundness implies (standard) soundness.
\begin{proposition}\label{prop:dpcptopcp}
  If CSAT\ has a $k$-prover PCP with distributional soundness error
  $\delta$, then \CSAT\ has a $k$-prover PCP with (standard) soundness
  error $\delta$. Furthermore, all other parameters (randomness,
  answer size, alphabet, perfect completeness) are identical.
\end{proposition}
\begin{proof}
  Suppose there exists a circuit $\Phi$ and a proof $\Pi$
such that $\Pr_R\left[V^\Pi(\Phi;R) = \acc\right] > \delta$. Then, by the distributional soundness
  property it follows that there exists at least one local accepting
  view which is a projection of a valid proof. In particular, there
  exists a valid proof which implies $\Phi \in \CSAT$.
\end{proof}

\subsection{PCP decoders}

We now present a variant of PCP verifiers, called PCP decoders,
introduced by Dinur and Harsha~\cite{DinurH2013}.  PCP decoders, as
the name suggests, have the additional property that they not only
{\em locally check} the PCP proof $\Pi$, but can also {\em locally decode}
symbols of an encoding of the original NP witness from the PCP
proof $\Pi$. PCP decoders are implicit in many previous constructions
of PCPs with small soundness error and were first explicitly defined under the name of local-readers by
Dinur~\etal~\cite{DinurFKRS2011}, as locally-decode-or-reject-codes (LDRC) by Moshkovitz and
Raz~\cite{MoshkovitzR2010b} and as decodable PCPs by Dinur and
Harsha~\cite{DinurH2013}.
As in the case of
PCP verifiers, our PCP decoders will be projection PCP decoders.
\newpage

\begin{definition}[PCP decoder]\label{def:pcpdecoder} \
\begin{itemize}
\item A {\em $k$-prover $l$-answer projection PCP decoder} over
  alphabet $\F$ and encoding length $t$ is a probabilistic-time
  algorithm $\D$ that on input $(\Phi, F)$ of size $N$, a random input
  string $R$ of $r(N)$ random bits and an additional input index $j \in [t]$,
  where $\Phi: \F^{n} \to \set{0,1}$ is a predicate, and
  $F=(F_1,\dots,F_{l-1})$ a list of $l-1$ functions $F_i:\F^n \to \F$,
  generates a tuple $(q, \phi, g,\func)$ where $q =
  (u,v_1\dots,v_{k-1})$ is a vector of $k$ queries, $\phi: \F^{m}\to
  \set{0,1}$ is a predicate, $g = (g_1,\dots,g_{k-1})$ is a list of
  $k-1$ functions $g_i : \F^{m} \to \F$ and $\func = (\func_0,\dots,
  \func_{l-1})$ is a list of $l$ functions $\func_i : \F^m \to \F$ such
  that the size of the tuple $(\phi, g, \func)$ is at most $s(N)$.
\item We write $(q,\phi,g,\func) = \D(\Phi, F; R,j)$ to denote the
  query-predicate-functions tuple output by the decoder $\D$ on input pair
  $(\Phi, F)$, random input $R$ and input index $j$.

\item We think of $\D$ as a probabilistic oracle machine that on input
  $(\Phi,F;R,j)$ queries $k$ provers  $\Pi = (A, B_1, \dots,
  B_{k-1})$ at positions $q = (u,v_1,\dots,v_{k-1})$ respectively to
  receive the answers $\Pi|_q := $\\ $(A(u),B_1(v_1),\dots,
  B_{k-1}(v_{k-1}))$ $\in
  \F^m \times \F^{k-1}$, then checks if $\phi(A(u)) = 1$ and $g_i(A(u))
  = B_i(v_i)$ for all $i \in \{1,\dots, k-1\}$ and if these tests pass
  outputs the $l$-tuple
  $(\func_0(A(u)),\func_1(A(u)),\dots,\func_{l-1}(A(u))) \in \F^{l+1}$
  and otherwise outputs $\bot$.
\item Given $k$ provers $\Pi = (A, B_1,\dots,B_{k-1})$,  we will sometimes
  collectively refer to them as the ``proof $\Pi$'' . Furthermore, we refer to $A$ as the large
  prover and the $B_i$'s as the projection provers. We call $\Pi|_q$ the local
    view of the provers $\Pi$ on queries $q$ and denote by
  $\D^{\Pi}(\Phi,F; R,j)$ the output of the decoder $\D$ on input
  $(\Phi, F; R,j)$ when interacting with the $k$ provers $\Pi$. Note that the
  output is an element of $\F^{l+1} \union \set{\bot}$.
\item We call $N$ the {\em input size}, $k$ the {\em number of
    provers}, $l$ the {\em number of answers}, $r(N)$ the {\em
    randomness complexity}, and $s(N)$ the {\em answer size} of the
  decoder $\D$.
\end{itemize}
\end{definition}

We now equip the above defined PCP decoders with the new notion of
soundness, distributional
soundness.
%
%
We find it convenient (and sufficient) to define decodable PCPs only
for predicates and function tuples which have an algebraic structure
over the underlying alphabet, which is the field $\F$. In other words,
both the input tuple $(\Phi,F)$ and output tuple $(\phi,g,\func)$ have
the property that the predicates $\Phi, \phi$ and the functions $F,
\func, g$ are specified as arithmetic circuits over $\F$.  For the
above reasons, we define dPCPs for $\ACSAT$ (see \lref[Definition]{def:acsat}).

\begin{definition}[decodable PCPs with distributional soundness]\label{def:dpcpdistsound}For
  $\delta \in (0,1)$ and a code $E:\F^n\to \F^t$, a $k$-prover $l$-answer projection PCP decoder $\D$ is a {\em
    $k$-prover $l$-answer decodable probabilistically checkable proof
    system} for $\ACSAT_\F$ with respect to encoding $E$
  with distributional soundness error $\delta$ if the following
  properties hold for every input pair $(\Phi,F)$:
\begin{description}
\item[Perfect Completeness:] For every $x \in \SAT(\Phi)$, there exist
  $k$ provers
$\Pi = (A, B_1,\dots, B_{k-1})$ such that the PCP decoder $\D$ when
interacting with provers $\Pi$ outputs
$(E(x)_j, F_1(x),\dots, F_{l-1}(x))$ for every random input $R$ and index $j$. I.e.,
$$\Pr_{R,j}
\left[\D^\Pi(\Phi,F;R,j)=(E(x)_j, F_1(x), \dots, F_{l-1}(x))\right] =
1.$$
In other words, the decoder $\D$ on input $(\Phi, F; R,j)$ outputs the
$j$-th symbol of the encoding $E(x)$ and the tuple $F$ evaluated at
$x$.  In this case, we say that $\Pi$ is a {\em valid proof} for the statement $x\in
\SAT(\Phi)$.
\item[Distributional Soundness:] For any set of provers $\Pi =
  (A,B_1,\dots,B_{k-1})$ there exists an idealized pair of functions $\Vix(R)$ and
  $\ViPi(R)$ defined for every random string $R$ such that the following holds.
\begin{itemize}
\item For every random string $R$,
  $\ViPi(R)$ is either a valid proof for the statement $\Vix(R)\in
  SAT(\Phi)$, or $\ViPi(R)=\bot$.
\item For every $j\in\set{1,\ldots,t}$, with probability at least
  $1-\delta$ over the choice of the random string $R$, the local view
  of the provers $\Pi$ completely agrees with the local view of the
  provers $\ViPi(R)$ or is a rejecting local view. In other words,
$$\forall j \in [t],\qquad \Pr_R \left[ \D^\Pi(\Phi,F; R,j) = \bot
  \text{ or } \Pi|_q = \ViPi(R)|_q \right] \geq 1-\delta,$$
where $q$ is the query vector generated by the PCP decoder $\D$ on
input $(\Phi, F;R,j)$.

When the local views agree, i.e. $\Pi|_q = \ViPi(R)|_q $, we say that $\ViPi$ is {\em successful} (in explaining the success of $\Pi$). In this case we have that $\D^\Pi(\Phi, F;R,j)
= (E(\Vix(R))_j,F(\Vix(R)))$.
\end{itemize}
\end{description}
We then say that $\D$ is a $k$-prover $l$-answer PCP decoder for
$\ACSAT$ with respect to encoding $E$ with perfect completeness and distributional soundness
error $\delta$.
\end{definition}


\begin{remark}\label{rem:uniform}The above definition is a {\em non-uniform} one in the
  sense that it is defined for a particular choice of input lengths
  $n$, $N$, size of field $\F$ and encoding $E:\F^n \to \F^t$. A {\em
    uniform} version of the above definition can be obtained as
  follows: there exists a polynomial time uniform procedure that on input $n, N$ (both
  in unary), the field $\F$ (specified by a prime number and
  an irreducible polynomial) and the encoding $E$ (specified by the
  generator matrix) outputs the PCP decoder algorithm. We note that our
  construction satisfies this stronger uniform property.
\end{remark}

As in previous works~\cite{MoshkovitzR2010b,DinurH2013}, dPCPs imply
PCPs with similar parameters
\begin{proposition}\label{prop:dpcptopcp}
If \ACSAT\ has a $k$-prover dPCP with distributional soundness error
$\delta$, then \ACSAT\  has a $k$-prover PCP with distributional soundness
error $\delta$. Furthermore, all other parameters (randomness, answer
size, alphabet, perfect completeness) are identical.
\end{proposition}

We conclude this section highlighting the differences/similarities
between the above notion of PCP decoders/dPCPs with that of Dinur and
Harsha~\cite{DinurH2013} besides the obvious difference in the soundness criterion.
\begin{remark}\

\begin{itemize}
\item The above definition of PCP decoders is a
generalization of the corresponding definition of Dinur and
Harsha~\cite{DinurH2013} to the multi-prover ($k>2$) setting. Since our PCP
verifiers are multi-prover verifiers and not just 2-prover verifiers, so
are our PCP decoders. Thus, in our notation, the PCP decoders of
\cite{DinurH2013} are $2$-prover $1$-answer projection PCP decoders.
\item The above defined PCP decoders locally decode symbols of some
  pre-specified encoding $E$ of the NP-witness. The PCP decoders of
  Dinur and Harsha~\cite{DinurH2013} is a special case of this when the encoding $E$
  is the identity encoding. However as we will see in
  the next section, it will be convenient to work with encodings which
  have good distance. In particular, the dPCP composition (considered
  in this paper) requires the encoding of the ``inner'' PCP decoder to
  have good distance.
\end{itemize}
\end{remark}

\section{Composition}\label{sec:composition}

In this section, we describe how to compose two PCP
decoders. Informally speaking, an ``outer'' PCP decoder $\D_\outr$ can
be composed with an ``inner'' PCP decoder $\D_\inn$ if the answer size of the
outer PCP decoder matches the input size of the inner PCP decoder and
the number of answers of the inner PCP decoder is the sum of the
number of answers of the outer PCP decoder and the number of provers of
the outer PCP decoder.

We begin with an informal description of the composition procedure. It might be useful to read this description while looking at \lref[Figure]{fig:comp}, in which there are three dPCPs: the inner, the outer, and the composed. We depict each dPCP as a bi-partite (or 4-partite) label-cover-like graph whose vertices correspond to proof locations, and whose (hyper-)edges correspond to local views of the PCP decoder.
\begin{figure*}[Composition]
\begin{center}\includegraphics[width=0.9\textwidth]{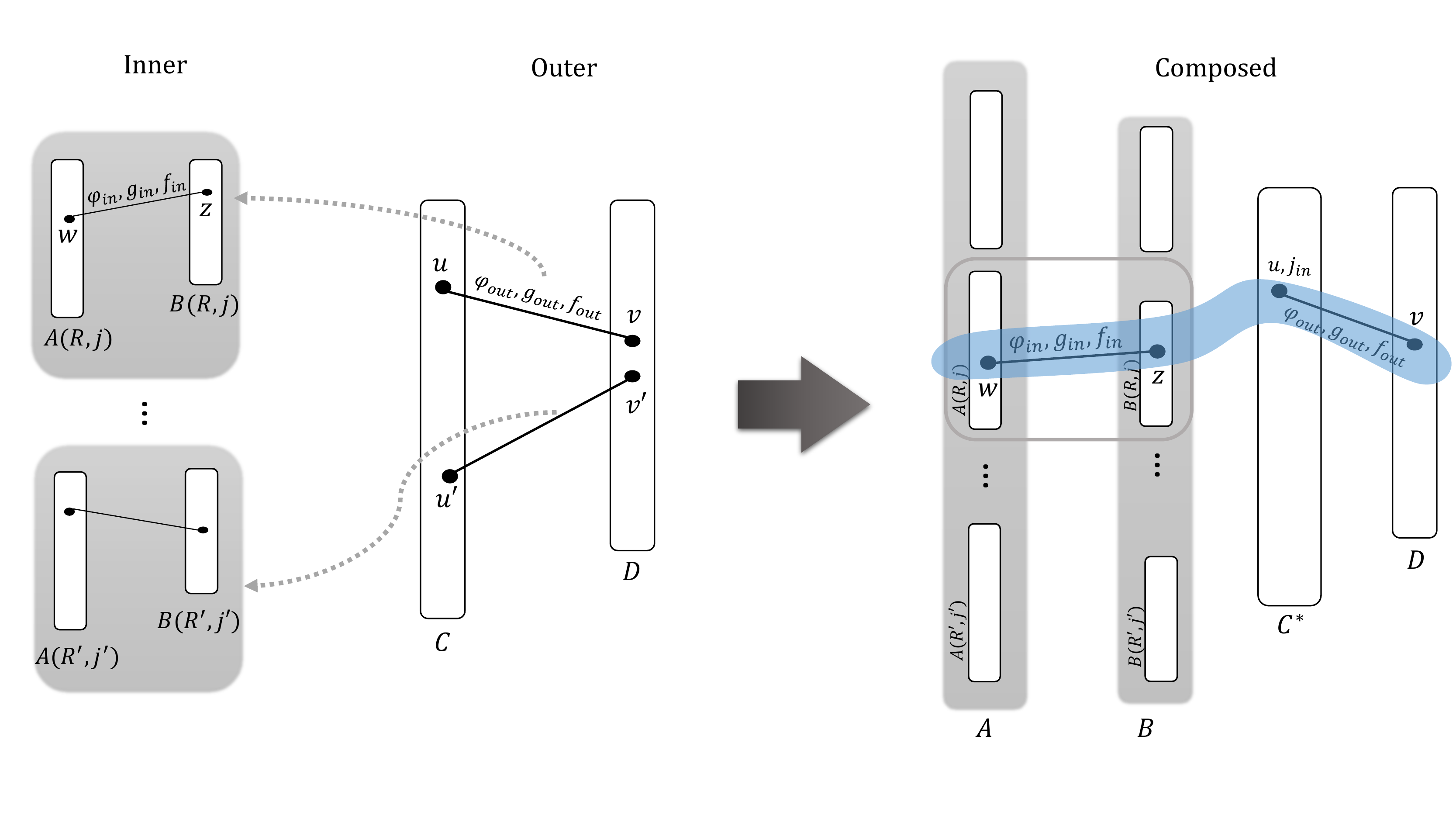}
\caption{Composition of two 2-prover dPCPs $\D_\outr$ and $\D_\inn$ to
  yield composed dPCP $\D_\comp$. Note that $\D_\inn,\D_\outr$ make two queries
  each, and $\D_\comp$ makes four queries: $w,z,(u,j_\inn),v$ to
  $A,B,C^*,D$ respectively.}
\end{center}\label{fig:comp}
\end{figure*}
The main goal of composition is to reduce the answer
size of the outer PCP decoder. By this we are referring to the answer size of the large prover; as it is always possible to reduce the answer size of projection provers at negligible cost. For simplicity, let us assume that each of the inner and the outer PCP decoders use only two provers. The inner PCP decoder $\D_\inn$ interacts with provers $A$ and $B$, and the outer PCP decoder $\D_\outr$ interacts with provers $C$ and $D$. The composed PCP decoder $\D_\comp$
works as follows: On input $\left(\Phi,F\right)$, $\D_\comp$ simulates $\D_\outr$ to obtain the tuple $\left(q_\outr, \phi_\outr, g_\outr,\func_\outr\right)$. Letting $q_\outr=(u,v)$ we picture this as an edge in the bipartite graph of the outer dPCP, and we label this edge with $\left(\phi_\outr, g_\outr, \func_\outr\right)$.
In its normal running
$\D_\outr$ generates queries $q_\outr=(u,v)$ and queries $C$ on $u$ and $D$ on $v$. It then checks that $\varphi_\outr(C(u))=0$ and $g_\outr(C(u)) = D(v)$ and if so it outputs $f_\outr(C(u))$.

However, the answer $C(u)$
is too large for $\D_\comp$, and we would like to use the inner PCP decoder $\D_\inn$ to replace
querying $C$ directly, reducing the
answer size  at the cost of a few extra queries. For this
purpose, the composed PCP decoder $\D_\comp$ now simulates the inner
PCP decoder $\D_\inn$ on input $\left(\phi_\outr,\left( g_\outr,\func_\outr\right)\right)$ to
generate the tuple $\left(q_\inn,\phi_\inn,g_\inn,\func_\inn\right)$. The
composed PCP decoder $\D_\comp$ then queries the inner provers
$A,B$ on queries $q_\inn=(w,z)$ to obtain the answers $\alpha = A(w)$ and $\beta = B(z)$. It then performs the
projection tests $g_\inn$ of the inner PCP decoder $\D_\inn$ and
produces its output $\func_\inn\left(\alpha\right)$. These answers are then used
to both perform the projection test of the outer PCP decoder as well
as produce the required output of the outer PCP decoder.

As usual in composition, we need to enforce consistency between the
different invocations of $\D_\inn$. The input for $\D_\inn$, namely $\left(q_\inn,\phi_\inn,g_\inn,\func_\inn\right)$, is generated using $\D_\outr$'s randomness, namely $R$ and $j$. The provers $A$ and $B$ must be told this input because they need to know what they are supposed to prove. Thus $A$ and $B$ are actually aggregates of prover-pairs $A(R,j),B(R,j)$ ranging over all possible $R,j$. There is a possibility that they could ``cheat'' by outputting a different answer for the same outer question, depending on $R$ and $j$. In particular, think of two outer query pairs $(u,v_1)$ and $(u,v_2)$ generated by two different random strings $R_1,j_1$ and $R_2,j_2$. We need to ensure that both invocations of $\D_\inn$ are consistent with the same answer $C(u)$.

We address this issue using the decoding feature of the inner PCP decoder $\D_\inn$. We replace the outer prover $C$ by a prover $C^*$, which we call the consistency prover. This prover is supposed to hold an encoding, via $E_\inn$, of the outer prover $C$.
The composed PCP decoder $\D_\comp$ expects
the inner PCP decoder $\D_\inn$ to decode a random symbol in this
encoding (i.e., in
$E_\inn\left(C(u)\right)$). This decoded value is then checked
against the consistency prover $C^*$, which unlike the inner provers
is not informed of the outer randomness $R$.
In all, the queries of $\D_\comp$ are $w,z,(u,j_\inn),v$ to $A,B,C^*,D$ respectively.

This additional consistency query
helps us get around the above mentioned issue at a small additional
cost of $\eta_\inn$ in the soundness error, where $\eta_\inn$ is
the agreement parameter of the encoding $E_\inn$ (see
\lref[Remark]{rem:agreement}).

It can be shown that this consistency
query ensures that the distributional soundness error of the composed
decoder is at most the sum of the distributional soundness errors of
the outer and inner PCP decoders and the agreement parameter of the
encoding $E_\inn$. Previous soundness analyses using list-decoding
soundness typically involved a $L_\inn$-fold multiplicative blowup in the soundness
error $\delta_\outr$ of the outer PCP decoder (i.e., $\delta_\comp
\geq L_\inn \cdot \delta_\outr$) where $L_\inn$ is the list-size of
the inner PCP decoder. Distributional soundness has the advantage of
getting rid of this $L_\inn$-fold blowup at the cost an additional
$\eta_\inn$ additive error.

The above description easily generalizes to $k>2$ by replacing $B$ by $B_1,\ldots,B_{k_\inn-1}$ and $D$ by $D_1,\ldots,D_{k_\outr-1}$.

As in the case of the definition of decodable PCPs, we find it
sufficient to describe composition of algebraic dPCPs and not general
dPCPs.

\begin{theorem}[Composition Theorem]\label{thm:comp} Let $\F$ be a
  finite field. Suppose that
  $N_\outr,N_\inn,r_\outr,r_\inn,s_\outr,$ $s_\inn,n_\outr,n_\inn,$  $t_\outr,t_\inn,
  k_\outr,k_\inn, l_\outr,l_\inn \in \integers$, and
  $\delta_\outr,\delta_\inn,\eta_\inn \in [0,1]$ are such that
\begin{itemize}
\item \ACSAT\ has a $k_\outr$-prover $l_\outr$-answer decodable
  PCP $\D_\outr$ with respect to encoding $E_\outr:\F^{n_\outr} \to \F^{t_\outr}$
with randomness complexity $r_\outr$,
answer size $s_\outr$,
and distributional soundness error $\delta_\outr$ on inputs $\Phi$ of
size $N_\outr$,
\item $\ACSAT_\F$  has a $k_\inn$-prover $l_\inn$-answer decodable PCP
  $\D_\inn$ with respect to encoding $E_\inn:\F^{n_\inn}\to \F^{t_\inn}$
with randomness complexity $r_\inn$, answer size $s_\inn$, and
distributional soundness soundness error $\delta_\inn$ on inputs
$\phi$ of
size $N_\inn$
\item $s_\outr \leq n_\inn \leq N_\inn$,
\item $l_\inn = k_\outr + l_\outr$,
\item the inner encoding $E_\inn$ has agreement parameter $\eta_\inn$
\end{itemize}
Then, \ACSAT\ has a $k_\outr+k_\inn$-prover
$l_\outr$-answer dPCP, denoted $\D_\comp = \D_\outr\twocomp \D_\inn$, with respect to encoding
$E_\outr$ on inputs $\Phi$ of size $N$ with
\begin{itemize}
\item randomness complexity $r_\outr+r_\inn+ \log_2( t_\inn)$,
\item answer size $s_\inn$, and
\item distributional soundness error $\delta_\outr+ \delta_\inn+\eta_\inn.$
\end{itemize}
Furthermore, there exists a universal algorithm with black-box access
to $\D_\outr$ and $\D_\inn$ that can perform the actions of $\D_\comp$
(i.e. evaluating $\D_\comp\left(\Phi, F;R,j\right)$).
On inputs of size $N$, this algorithm runs in
time $N^c$ for a universal constant $c$, with one call to $\D_\outr$ on an input of size
$N$ and one call to $\D_\inn$ on an input of size $s_\outr$.
\end{theorem}

\begin{proof}

We will follow the following notation to describe the composed
decoder.

\paragraph{Provers of $\D_\comp$} Suppose the inner PCP
decoder $\D_\inn$ interacts with provers
$A,B_1,\dots,B_{k_\inn-1}$ (here $A$ is the large prover and $B_i$'s are the projection provers), and the outer PCP decoder interacts with provers $C,
D_1,\dots, D_{k_\outr-1}$ (here $C$ is the large
prover and $D_i$'s are the projection provers).

%
As mentioned in the informal description, the composed PCP decoder
$\D_\comp$ simulates $\D_\outr$ except that instead of
querying $C$, uses the inner PCP decoder $\D_\inn$ and an additional consistency prover
$C^*$. Thus, the provers for the composed PCP decoder $\D_\comp$ will
be the following: $A, B_1,\dots, B_{k_\inn-1}, C^*,
D_1,\dots, D_{k_\outr -1}$; the main prover being $A$
and the projection provers being the rest. As mentioned in the outline, for each choice of the outer randomness $R_\outr$ and index $j_\outr$ the inner PCP decoder $\D_\inn$
is simulated on a different input. Hence the corresponding inner provers for the
composed dPCP $\D_\comp$ (i.e., $A, B_1,\dots,
B_{k_\inn-1}$) are explicitly given the specification of the
outer randomness $R_\outr$ and index $j_\outr$ as part of their
queries. (Alternatively, one can think of $A$ and $B_i$ as an aggregate of separate provers $A(R_\outr,j_\outr)$ and $B_i(R_\outr,j_\outr)$ per $R_\outr,j_\outr$).

\paragraph{Randomness of $\D_\comp$} The randomness of $\D_\comp$ comes in three
parts: the randomness $R_\outr$ of $\D_\outr$, the randomness $R_\inn$
of $\D_\inn$ and a random index $j_\inn$ to perform the consistency
test. Thus, $R_\comp = \left(R_\outr, R_\inn, j_\inn\right)$.

\paragraph{Decoded Index of $\D_\comp$}  
The index $j_\comp$ being decoded
by $\D_\comp$ is passed as the index $j_\outr$ being decoded by
$\D_\outr$.


\paragraph{Indexing the answers of $\D_\comp$} Note that the number of answers $l_\inn$ of the inner PCP
decoder $\D_\inn$ is the sum of the number of answers $l_\outr$ of the
outer $\D_\outr$ and the number of provers $k_\outr$ of
outer $\D_\outr$. Thus, $\func_\inn$ is a list of $l_\outr+k_\outr$
functions. We will find it convenient to index the functions in
$\func_\inn$ with $\set{0} \union \left(\set{\outr} \times\{0,1,\dots,
l_\outr-1\}\right) \union \left(\set{\proj} \times \{1,\dots, k_\outr-1\}\right)$, such that $\func_{\inn,(\proj,i)}$, $i=1,\ldots,k_\outr-1$, are the answers to be compared with the outer projection provers, $\func_{\inn,0}$ is intended for the consistency test, and $\func_{\inn,(\outr,i)}$, $i=0,\ldots,\ell_\outr-1$, give the answers for the outer decoder.

With these conventions in place,
here is the description of the composed PCP decoder, $\D_\comp$:\\

\fbox{
\begin{minipage}{\textwidth}
$\D_\comp(\Phi, F;R_\outr, R_\inn, j_\inn)$:
\begin{itemize}
\item Input: $\left(\Phi, F\right)$
\item Random input string: $\left(R_\outr, R_\inn, j_\inn\right)$
\item Index to be decoded: $j_\outr$
\item Provers: $\Pi = \left(A, B_1,\dots, B_{k_\inn-1}, C^*,
D_1,\dots, D_{k_\outr -1}\right)$
\begin{enumerate}
\item Initial Computation:
\begin{enumerate}
\item {[Simulating $\D_\outr$]} Run $\D_\outr\left(\Phi,F; R_\outr, j_\outr\right)$ to obtain
$\left(q_\outr,\phi_\outr,g_\outr,\func_\outr\right)$.
\item {[Simulating $\D_\inn$]} Run $\D_\inn\left(\phi_\outr,\left(g_\outr,\func_\outr\right); R_\inn, j_\inn\right)$ to
obtain $\left(q_\inn,\phi_\inn,g_\inn,\func_\inn\right)$.
\end{enumerate}
\item Queries: Let $q_\outr = (u,v_1,\ldots,v_{k_\outr-1})$ and let $q_\inn = (w,z_1,\ldots,z_{k_\inn-1})$.
\begin{enumerate}
\item Send query $\left(R_\outr,j_\outr, w\right)$ to prover
  $A$ to obtain answer $\alpha = A(R_\outr,j_\outr,w)$.
\item For $i = 1, \ldots, k_{\inn-1}$, send query $\left(R_\outr,j_\outr, z_i\right)$ to
  prover $B_i$ to obtain answer $\beta_{i} = B_i(R_\outr,j_\outr,z_i)$.
\item Send query $\left(u, j_\inn\right)$ to prover $C^*$ to
  obtain answer $\gamma = C^*(u,j_\inn)$.
\item For $i =1 \dots, k_{\outr}-1$, send query $v_i$ to prover $D_i$ to
  obtain answer $\zeta_{i} = D_i(v_i)$.
\end{enumerate}
\item Checks:
\begin{enumerate}
\item\label{itm:chk-in-pred} {[Inner local predicate]} Check that
    $\phi_\inn\left(\alpha\right) = 1$.
\item\label{itm:chk-in-proj} {[Inner projection tests]} For $i = 1, \dots, k_{\inn -1}$, check that
  $g_{\inn,i}\left(\alpha\right)=\beta_{i}$.
\item\label{itm:chk-cons} {[Consistency test]} Check that
  $\func_{\inn,0}\left(\alpha\right)=\gamma$.
\item\label{itm:chk-out} {[Outer projection tests]} For $i = 1,\dots, k_{\outr-1}$, check
    that  $\func_{\inn, \left(\proj,i\right)}\left(\alpha\right)=\zeta_{i}.$
\end{enumerate}
\item {Output:} If all the checks in the above step pass, then return
  $\func_{\inn,\left(\outr,\cdot\right)}\left(\alpha\right)$ else
  return $\bot$.
\end{enumerate}
\end{itemize}
\end{minipage}
}

The claims about $\D_\comp$'s parameters (randomness complexity, answer
size, number of provers, number of answers) except completeness and
soundness error can be verified by inspection. Thus, we only need to
check completeness and soundness.

\paragraph{Completeness} Let $x \in SAT\left(\Phi\right)$. By the
completeness of outer $\D_\outr$, there exist provers
$\Pi^\outr = \left(C,D_1,\dots, D_{k_\outr-1}\right)$, such that
for all $\left(R_\outr, j_\outr\right)$ we have $$\D^{\Pi^\outr}_\outr\left(\Phi, F;
R_\outr,j_\outr\right) = \left(E_\outr\left(x\right)_{j_{\outr}}, F_1\left(x\right), \dots,
F_{l_\outr}\left(x\right)\right).$$
Fix any particular outer random string $R_\outr$ and index
$j_\outr$. Let $\D_\outr\left(\Phi, F; R_\outr, j_\outr\right) =
\left(q_\outr, \phi_\outr,g_\outr,\func_\outr\right)$. Since the outer
decoder $\D_\outr$ does not reject, we must have that
$y_{\left(R_\outr,j_\outr\right)} := C(u)$
satisfies $\phi_\outr$. In other words,
$y_{\left(R_\outr,j_\outr\right)} \in
\SAT\left(\phi_\outr\right)$. Now, by the completeness of the inner
$\D_\inn$, we have that for these $\left(R_\outr, j_\outr\right)$ there
exist provers $\Pi^{\inn}_{\left(R_\outr,j_\outr\right)} =
\left(A_{\left(R_\outr,j_\outr\right)},B_{\left(R_\outr,j_\outr\right),1}, \dots,
  B_{\left(R_\outr,j_\outr\right),k_\inn-1}\right)$ such that for
all $\left(R_\inn, j_\inn\right)$ we have
$$\D^{\Pi^{\inn}_{\left(R_\outr,j_\outr\right)}}_\inn\left(\phi_\outr,
\left(g_\outr,f_\outr\right); R_\inn, j_\inn\right) =
\left(E_\inn\left(y_{\left(R_\outr,j_\outr\right)}\right)_{j_\inn},
g_\outr\left(y_{\left(R_\outr,j_\outr\right)}\right),
\func_\outr\left(y_{\left(R_\outr,j_\outr\right)}\right)\right).$$

We are now ready to define the provers $$\Pi = \left(A, B_1,\dots,
B_{k_\inn-1}, C^*, D_1, \dots, D_{k_\outr-1}\right)$$ for
the composed decoder $\D_\comp$. As the name suggests, the projection
provers $D_i, i=1, \dots, k_{\outr}-1$ are exactly the same as
the outer projection provers in $\Pi^\outr$. The consistency prover
$C^*$ is defined by encoding $C(u)$ separately for each $u$ as
\[C^*\left(u, j_\inn\right) :=
E_\inn\left(C(u)\right)_{j_\inn}.\]
The projection provers
$B_i, i =1, \dots, k_{\inn} -1$ are defined as
$B_i\left(R_\outr,j_\outr, z_i \right) := B_{\left(R_\outr, j_\outr\right),i}\left(z_i\right)$. Finally, the large prover $A$ is defined as
$A\left(R_\outr, j_\outr,w\right) :=
A_{\left(R_\outr,j_\outr\right)}(w)$. It is
easy to check that according to this
definition of $\Pi$, for each
$\left(\left(R_\outr, R_\inn, j_\inn\right),
j_\outr\right)$ it holds that
$$\D^\Pi_\comp\left(\Phi, F; \left(R_\outr,R_\inn, j_\inn\right), j_\outr\right) =
\left(E_\outr\left(x\right)_{j_\outr}, F_1\left(x\right),\dots,
  F_{l_\outr}\left(x\right)\right).$$
This proves the completeness of $\D_\comp$.\qed

\paragraph{Distributional Soundness of $\D_\comp$} We prove the
following statement about the distributional soundness of $\D_\comp$.

\begin{lemma}\label{lem:distsound}
Suppose the outer PCP decoder $\D_\outr$ has distributional soundness error $\delta_\outr$
with respect to encoding $E_\outr$, and the inner PCP decoder $\D_\inn$ has distributional soundness
error $\delta_\inn$ with respect to encoding $E_\inn$, and suppose $E_\inn$ has
agreement parameter $\eta_\inn$ (see
\lref[Remark]{rem:agreement}). Then, the composed PCP decoder $\D_\comp
= \D_\outr \twocomp \D_\inn$ has distributional soundness error
$\delta_\comp \leq \delta_\outr+\delta_\inn+\eta_\inn$ with respect to encoding $E_\outr$.
\end{lemma}

\begin{proof}
Suppose $\D_\comp$ on input $\left(\Phi,F\right)$ interacts with provers
\[\Pi=\left(A,B_1,\ldots,B_{k_\inn-1},C^*,D_1,\ldots,D_{k_\outr-1}\right).\]
To prove soundness of $\D_\comp$, we need to construct for each
composed random
string $R_\comp:=\left(R_\outr, R_\inn, j_\inn\right)$, functions
$\Vix\left(R_\comp\right)$ and  $\ViPi\left(R_\comp\right)$ such
that $\ViPi\left(R_\comp\right)$ is either $\bot$ or a valid proof
for the statement $\Vix\left(R_\comp\right) \in
\SAT\left(\Phi\right)$, and for every $j_\outr$ we have
\begin{align*}
\Pr_{R_\comp}\left[\D_\comp^\Pi\left(\Phi, F;R_\comp, j_\outr\right) = \bot \text{ or } \Pi|_{q_\comp} =
\ViPi\left(R_\comp \right)|_{q_\comp} \right] \geq1-( \delta_\outr +
\delta_\inn + \eta_\inn),
\end{align*}
where $q_\comp$ is the query vector generated by $\D_\comp$.

\bigskip \noindent The construction of $\Vix\left(R_\comp\right)$ and
$\ViPi\left(R_\comp\right)$ relies on the soundness properties of
$\D_\outr$ and $\D_\inn$. We first locally-decode $C^*$ to obtain a distribution over outer provers $C$. We then use the distributional soundness of $\D_\outr$ to obtain an idealized outer proof $(\Vi C,\Vi D_i)$. We then use the soundness of $\D_\inn$ to obtain idealized inner proofs $(\Vi A, \Vi B_i)$.
\begin{description}
\item[Outer main prover $C_{j_\inn}$:] For each $j_\inn$, we
  define an outer main prover $C_{j_\inn}$ using the
  consistency prover $C^*$ as follows. Let $\tau = \eta_\inn/2$ be the agreement parameter of $E_\inn$, as in
  \lref[Remark]{rem:agreement}. For each query $u:= q_{\outr,0}$ to the
  outer main prover, $C_{j_\inn}\left(u\right)$ is defined to be the
  $j_{\inn}$-th entry of the $\tau$-local decoding of $C^*\left(u, \cdot\right)$
  as in \lref[Definition]{def:agreedist}  if well-defined and $\bot$
  otherwise.
  \inote{ removed the explicit def, see latex}

\item[Idealized outer pairs $\Vix^\outr_{j_\inn}\left(R_\outr\right)$
  and $\ViPi^\outr_{j_\inn}\left(R_\outr\right)$:] For each $j_\inn$ we have an
  outer proof \[\Pi^\outr_{j_\inn}=\left(C_{j_\inn},D_1,\ldots,D_{k_\outr-1}\right).\]
  Note that only the $C$ provers are different in the various outer
  provers $\Pi^\outr_{j_\inn}$ as we range over $j_\inn$.
  From the soundness of  $\D_\outr$ for every $\Pi^\outr_{j_\inn}$ there is
  an idealized pair $\Vix^\outr_{j_\inn}\left(R_\outr\right)$ and
  $\ViPi^\outr_{j_\inn}\left(R_\outr\right)=\left(\Vi C,\Vi
    D_1,\ldots,\Vi D_{k_\outr-1}\right)$%
  \footnote{The proofs $\left(\Vi C,\Vi
    D_1,\ldots,\Vi D_{k_\outr-1}\right)$ depend on $j_{\inn}$
  and $R_{\outr}$ so more formally could be written as $\left(\Vi C_{j_{\inn}}(R_\outr),\Vi
    D_{j_\inn,1}(R_\outr),\ldots,\Vi
    D_{j_\inn,k_\outr-1}(R_\outr)\right)$, but we will drop the
  indices for ease of readability.} %
  that ``explain'' its success.
  \inote{removed repeating of definition, see latex}

\item [Idealized inner pairs
  $\Vi y^\inn_{\left(R,j\right)}\left(R_\inn\right)$ and
  $\ViPi^\inn_{\left(R,j\right)} \left(R_\inn\right)$:]
  For every outer randomness
  $R$ and index $j$ let
\[
\Pi^\inn_{\left(R,j\right)} :=
\left(A\left(R,j,\cdot\right),B_1\left(R,j,\cdot\right),\ldots,B_{k_\inn-1}\left(R,j,\cdot\right)\right)\]
be the relevant part of the proof for $\D_\inn$.

Let $\left(q,\phi,g,\func\right) = \D_\outr\left(\Phi,F; R,
  j\right)$. When $R_\outr=R$ and $j_\outr=j$, the composed
decoder $\D_\comp$ simulates running $\D_\inn$ with input
$\left(\phi,\left(\g,\func\right); R_\inn, j_\inn \right)$ and with the proof
$\Pi^\inn_{\left(R,j\right)}$.


%
  For each $\Pi^\inn_{(R,j)}$, the soundness of $\D_\inn$ guarantees idealized prover pairs $\Vi y^\inn_{\left(R,j\right)}$ and
  $\ViPi^\inn_{\left(R,j\right)}$ (functions of $R_\inn$) that "explain" its success.
  \inote{again removed repeating the properties of $\ViPi$ from the def}
%

\end{description}

We are ready to define the idealized $(\Vi x,\ViPi)$ pairs for the composed
decoder $\D_\comp$.

\begin{description}
\item[Idealized composed pairs $\Vix\left(R_\comp\right)$ and $\ViPi\left(R_\comp
  \right)$:] Recall that $R_\comp$ is short for
$\left(R_\outr,R_\inn, j_\inn\right)$.
Define $\Vix\left(R_\comp\right)  :=\Vix^\outr_{j_\inn}\left(R_\outr\right)$. We then set
$\ViPi\left(R_\comp\right)$ as follows:  If
$\ViPi^\outr_{j_\inn}\left(R_\outr\right)$ is successful, in
particular $C_{j_\inn}(u) = \Vi C(u)$, set $\ViPi\left(R_\comp\right)$ to be the set of provers  \\
$\left({\Vi A},
{\Vi B}_1,\ldots,{\Vi B}_{k_\inn-1}, {\Vi C}^*,
{\Vi D}_1,\ldots,{\Vi D}_{k_\outr-1}\right)$ defined next.
\begin{itemize}
\item The outer projection provers ${\Vi D}_i$ are defined to be the
  same as in $\ViPi^\outr_{j_\inn}\left(R_\outr\right)$.
\item Let ${\Vi C}$ be the main prover in
  $\ViPi^\outr_{j_\inn}\left(R_\outr\right)$. We define ${\Vi C}^*$ as
  follows:
  \[{\Vi C}^*\left(u,j\right) :=E_\inn\left({\Vi C}(u)\right)_j.\]
\item For any pair $(R,j)$ where $R\in\set{0,1}^{r_\outr}$ and
  $j\in[t_\outr]$, we define ${\Vi A}\left(R,j,\cdot\right)$ and
  ${\Vi B}_i\left(R,j,\cdot\right)$ (for
  $i=1,\ldots,k_\inn-1$) as follows. Denote $\left(q,\phi,g,\func\right) = \D_\outr\left(\Phi,F;
      R, j\right)$ and suppose $q=(u,v_1,\ldots,v_{k_\outr-1})$.  If
    $\Vi y^\inn_{(R,j)}(R_\inn)= {\Vi C}(u)$ and also $\ViPi^\inn_{\left(R,j\right)}
  \left(R_\inn\right)$ is successful, we set ${\Vi
    A}\left(R,j,\cdot\right)$ and the provers ${\Vi
    B}_i\left(R,j,\cdot\right)$ to be the main prover
  and projection provers in
  $\ViPi^\inn_{\left(R,j\right)}\left(R_\inn\right)$
  respectively.

  Otherwise,
  \inote{need to add a remark explaining this subtlety}
  if $\Vi y^\inn_{(R,j)}(R_\inn) \neq  {\Vi
    C}(u)$ or $\ViPi^\inn_{\left(R,j\right)}
  \left(R_\inn\right)$ is not successful, we define  ${\Vi A}\left(R,j,\cdot\right)$
  and ${\Vi B}_i\left(R,j,\cdot\right)$ by letting them be some valid proofs
  for the statement $ {\Vi
    C}(u)\in SAT \left(\phi\right)$ (Note that $ {\Vi
    C}(u) $ satisfies $\phi$
  since $\ViPi^\outr_{j_\inn}\left(R_\outr\right)$ is successful).
\remove{
\item
We define the main prover ${\Vi A}^\inn\left(R_\outr,j_\outr,\cdot\right)$ to be equal to the
 main prover of $\ViPi^\inn_{\left(R_\outr,j_\outr\right)}\left(R_\inn,j_\inn\right)$ and we define
  the projection prover ${\Vi B}_i^\inn\left(R_\outr,j_\outr,\cdot\right)$ to be the same as the $i$'th
  projection prover in $\ViPi^\inn_{\left(R_\outr,j_\outr\right)}\left(R_\inn,j_\inn\right)$.
\item For any pair $\left(R,j\right)\neq
  \left(R_\outr,j_\outr\right)$, we let $\left(q,\phi,g,f\right)=\D_\outr\left(\Phi,F;R,j\right)$, and set
  $\Vi y:={\Vi A}^\outr\left(q_{\outr,0}\right)$. We define the answers ${\Vi A}^\inn\left(R,j,\cdot\right)$
  and ${\Vi B}_i^\inn\left(R,j,\cdot\right)$ by letting them be valid proofs
  for the statement $y\in SAT \left(\phi\right)$ (Note that $y$ satisfies $\phi$
  since $\ViPi^\outr_{j_\inn}\left(R_\outr,j_\outr\right)$ is a valid
  proof).
}
\end{itemize}
If either $\ViPi^\outr_{j_\inn}\left(R_\outr\right)$ is not successful
or any of the intermediate objects in the above definition are $\bot$,
then we set $\ViPi\left(R_\comp\right)$ to $\bot$.
\end{description}

It remains to show that the pair $\Vix\left(R_\comp\right)$ and $\ViPi\left(R_\comp\right)$
  has the desired properties. \phnote{Needs more explanation} It
  follows  by inspection of the definition of
  $\ViPi\left(R_\comp\right)$ that whenever it is not
  $\bot$, it is a valid proof
  of the statement $\Vix\left(R_\comp\right)\in SAT\left(\Phi\right)$  and agrees with the
  local view \gnote{elaborate!} of $\Pi$ on input $\left(\Phi, F;R_\comp, j_\outr\right)$.

So it remains to show that for every $j_\outr$,
\begin{align*}
\Pr_{R_\comp}\Bigl[\ViPi\left(R_\comp\right)=\bot
\ve
\D_\comp^{\Pi}\parenth{\Phi,F;R_\comp,j_\outr}\neq\bot\Bigr]
\leq \delta_\outr+\delta_\inn+\eta_\inn.
\end{align*}
We partition the above event intro three parts according to the
highest indexed condition among the following three conditions that does not hold --- one of them must not hold for
$\ViPi\left(R_\comp\right)$ to be equal to $\bot$.

\begin{enumerate}
\item \label{cond1} $\ViPi^\outr_{j_\inn}\left(R_\outr\right)$ is successful, in particular $C_{j_\inn}(u) = \Vi C(u)$.
\item\label{cond2} $\Vi y^\inn_{\left(R_\outr,j_\outr\right)} \left(R_\inn\right)=\Vi C_{j_\inn}(u)$.
\item \label{cond3} $\ViPi^\inn_{\left(R_\outr,j_\outr\right)}\left(R_\inn\right)$ is successful.
\end{enumerate}

We separately bound the probability of each event in this partition.
  \begin{itemize}

   \item We bound the probability that \lref[Condition]{cond3} does not
    hold, namely  that
    $\ViPi^\inn_{\left(R_\outr,j_\outr\right)}\left(R_\inn\right)$ is
    not successful,  and yet $\D_\comp$ does not reject. If $\D_\comp$ doesn't reject then in particular checks \lref[Check]{itm:chk-in-pred} and \lref[Check]{itm:chk-in-proj} pass,
    which  means that
    $\D_\inn^{\Pi^\inn_{(R_\outr,j_\outr)}}\left(\phi_\outr,\left(g_\outr,\func_\outr\right);
      R_\inn, j_\inn\right) \neq\bot$. But the soundness of $\D_\inn$ implies that the probability over
    the choice of $R_\inn$ that this occurs and yet
    $\ViPi^\inn_{\left(R_\outr,j_\outr\right)}\left(R_\inn\right)$ is
    not successful is bounded by $\delta_\inn$.

  \item Now we bound the probability that \lref[Condition]{cond3} holds,
    \lref[Condition]{cond2} does {\em not} hold, and yet $\D_\comp$ does not
    reject. When \lref[Condition]{cond3} holds, the output of the
    $\D_\inn$ simulation for the encoding is
    $E_\inn\left(y^\inn_{\left(R_\outr,j_\outr\right)}\left(R_\inn\right)\right)_{j_\inn}$. It is thus
    enough to bound the probability that
    $\Vi y^\inn_{\left(R_\outr,j_\outr\right)}\left(R_\inn\right)\neq
    \Vi C(u)$ and yet
    $E_\inn\left(\Vi y^\inn_{\left(R_\outr,j_\outr\right)}\left(R_\inn\right)\right)_{j_\inn}=C^*\left(u,j_\inn\right)$, i.e. \lref[Check]{itm:chk-cons} passes.
Since by definition $C_{j_\inn}(u)$ is the $\tau$-local decoding of
$C^*\left(u,\cdot\right)$ at position $j_\inn$,
\lref[Claim]{claim:ambiguous} and \lref[Remark]{rem:agreement} imply that the probability of this event
over the choice of $j_\inn$ is bounded by the agreement parameter $\eta_\inn$.

\remove{
    The above inequality can occur either because
    $y^\inn\left(j_\outr,R^\outr,R^\inn\right)\right)$ is not $\tau$-admissible with
    respect to $A^*\left(q^\outr_{A^\outr},\cdot\right)$, in which case the
    latter equality holds with probability at most $\tau$ over the
    random selection of $j_\inn$, or because
    $y^\inn\left(j_\outr,R^\outr,R^\inn\right)\right)$ is $\tau$-admissible and yet
    $E^\inn\left(y^\inn\left(j_\outr,R^\outr,R^\inn\right)\right)\right)$ is ambiguous at $j_\inn$
    with respect to $A^*\left(q^\outr_{A^\outr},\cdot\right)$. However according
    to \lref[Claim]{claim:ambiguous} this happens with probability at
    most $4\eta/\tau^2$ over the choice of $j_\inn$.}

  \item It remains to bound the probability that
    Conditions~\ref{cond3} and~\ref{cond2} hold but \lref[Condition]{cond1}
    does not, and yet $\D_\comp$ does not reject. When
    Condition~\ref{cond2} and~\ref{cond3} hold it means that
    $\Vi C\left(u\right)=
    \Vi y^\inn_{\left(R_\outr,j_\outr\right)}\left(R_\inn\right)\in SAT\left(\phi_\outr\right)$ and the output of
    the simulated $\D_\inn$ computed by $\D_\comp$
    is
    \[
    f_\inn\left(A\left(R_\outr,j_\outr,w \right)\right)=
    \parenth{
      \parenth{ E_\inn\left(\Vi C\left(u\right)\right)_{j_\inn},
        g_\outr\left( \Vi C\left(u\right)\right),
        f_\outr\left( \Vi C\left(u\right) \right)}}.
    \]
    If $\D_\comp$ does not reject it means that the values of $g_\outr\left(
    \Vi C(u) \right)$ match the ones obtained from
    the outer projection provers, to which it is compared in \lref[Check]{itm:chk-out}. But these are also the values used by $\D_\outr$ when it
    is run on input $\left(\Phi, F; R_\outr,j_\outr\right) $ with the proof $\ViPi^\outr_{j_\inn}$, which means that
    $\D_\outr^{\ViPi^\outr_{j_\inn}} \left(\Phi, F;
    R_\outr,j_\outr\right)\neq \bot$. But the probability over the choice of
    $R_\outr$ that this happens while
    \lref[Condition]{cond1} fails is bounded by $\delta_\outr$, the distributional
    soundness error of $\D_\outr$.
  \end{itemize}
This proves the distributional soundness of $\D_\comp$.

\end{proof}
This completes the proof of the \lref[Composition Theorem]{thm:comp}.
\end{proof}

\section{Proof of Main Theorem}\label{sec:together}

\begin{theorem}[Main Construct]\label{thm:maincons} Every language $L$
  in $\NP$ has a $O(\lg \lg N/\lg\lg\lg N)$-prover projective PCP with
  the following parameters. On input a Boolean predicate/circuit $\Phi$ of size $N$, the PCP has
\begin{itemize}
\item randomness complexity $O(\lg N)$,
\item query complexity $O(\lg\lg N/\lg\lg\lg N)$,
\item answer size $O(\lg N/\poly\lg\lg N)$,
\item perfect completeness, and
\item soundness error $N^{1/(\lg \lg N)^{\Omega(1)}}$.
\end{itemize}
\end{theorem}
The PCP with inverse polynomial soundness error stated in \lref[Main
Theorem]{thm:main} is obtained by sequentially repeating the above PCP
$\poly(\lg\lg N)$ times in a randomness efficient manner.

\subsection{Building Blocks}

The two building blocks, we need for our construction, are two decodable
PCP based on the Reed-Muller code and the 
Hadamard code respectively. The constructions of both these objects is standard
given the requirements of the dPCP. These PCPs are based on two
encodings the low-degree encoding $\LDE$ and the quadratic Hadamard
encoding $\QF$ respectively. The definition of these codes is given in the next
section (\lref[\S]{sec:components}). For the purpose of this section,
it suffices that these are error correcting codes with very good
distance.

\begin{theorem}[Reed-Muller based dPCP]\label{thm:sumcheck} For any finite field $\F$,
  and parameter $h$ such that $1 < h < |\F|^{0.1}$ and any $\ell>0$, there
    is a 2-prover $\ell+1$-answer decodable PCP $\D$ with respect
    to the
    encoding $\LDE_{\F,h}$ for the language $\ACSAT_\F$ with the following parameters:
    On inputs (i) a predicate $\Phi:\F^{n}\to\bits$ and (ii) functions
    $F_1,\ldots,F_{\ell}:\F^{n}\to\F$ given by arithmetic circuits over $\F$ whose total size is $N$, the dPCP $\D$ has (let $m = \log N/\log h$),
\begin{itemize}
\item randomness complexity $O(\log {N}+ m\log |\F|) = O(m\log\card\F)$\phnote{???},
\item answer size $s,s'=O(m(m+\ell))$,
\item and distributional soundness error $1/|\F|^{0.1}$.
\end{itemize}
\end{theorem}

\begin{theorem}[Hadamard based dPCP]\label{thm:hadamard} For any finite field $\F$, and
  any $\ell>0$, there is a 2-prover $\ell+1$-answer decodable PCP $\D_{\QF,\F}$
  with respect to the encoding $\QF_\F$ for the language $\ACSAT_\F$
  with the following parameters: On inputs (i) a predicate
  $\Phi:\F^{n}\to\bits$ and (ii) functions
  $F_1,\ldots,F_{\ell}:\F^{n}\to\F$ given by arithmetic circuits over $\F$ whose total size is $N$, the dPCP $\D_{\QF,\F}$ has
\begin{itemize}
\item randomness complexity $O(N^2 \log |\F|)$,
\item answer size $s,s'=O(\ell)$,
\item perfect completeness, and
\item distributional soundness error $\le 1/|\F|^{0.1}$.
\end{itemize}
\end{theorem}

These theorems are proved in \lref[Section]{sec:components}.

\subsection{Putting it together (Proof of {\lref[Theorem]{thm:maincons}})}

By NP-completeness of \CSAT\ it suffices to prove
\lref[Theorem]{thm:maincons} for \CSAT. Let $\Psi$ be an instance of
$\CSAT$ and let $N$ denote its size.  Let $\epsilon = 20\lg\lg \lg N/9\lg \lg N$ be a
parameter\footnote{In the construction, setting $\epsilon = 20\lg\lg\lg
  N/9\lg \lg N$ will prove the $poly\lg\lg N$-query PCP with inverse
  polynomially soundness error (as stated in the main theorem). It is
  to be noted that setting $\epsilon$ to a constant in $(0,1)$ will
  recover the DFKRS PCP.}. Note that $(\lg N)^\epsilon = (\lg \lg
N)^{20/9}$.  Let $M=2^{(\lg N)^{1-\epsilon}}=N^{1/(\lg\lg
  N)^{20/9}}$. Choose a prime number $p \in (M,2M)$\footnote{Since the
  procedure is allowed to run in polynomial time in $N$, it has enough
  time to examine every number in the range $(M,2M)$ and check if it
  is prime or not.} and let $\F=GF(p)$ be the
finite field of size $p$, which we fix for the rest of the proof. We
may assume wlog. that the predicate $\Psi$ has only AND and NOT
gates. Given this, we can arithmetize $\Phi$ to obtain an arithmetic
circuit $\Phi$ over $\F$ by replacing AND gates by multiplication gates and
NOT gates by $1-x$ gates. Thus, we can view the $N$-sized predicate
$\Phi$ as an $N$-sized arithmetic circuit over the field $\F$.

We construct a PCP for $\Psi$ with the required parameters by
constructing a dPCP for $\Phi$ with respect to the encoding
$\LDE_{\F,h_0}$ for some suitable choice of $h_0$. This dPCP is in
turn constructed by composing a sequence of dPCPs each with smaller
and smaller answer size. Each dPCP in the sequence will be obtained by
composing the prior dPCP (used as an outer dPCP) with an adequate
inner dPCP. The outermost dPCP as well as the inner dPCP in all but
the last step of composition will be obtained from
\lref[Theorem]{thm:sumcheck} by various instantiations of the
parameter $h$. The innermost dPCP used in the final stage of the
composition will be the dPCP obtained
from \lref[Theorem]{thm:hadamard}.\\

\noindent {\sf Stage I:} Let $n_0= N$ and $h_0 = |\F|^{0.1}=2^{0.1
  (\lg N)^{1-\epsilon}} =N^{0.1/(\lg\lg N)^{20/9}}$.  For this choice of $n_0,
h_0$ and $\F$ and $l_0=0$, let $\D^{(0)} := \D_0$ be the dPCP obtained
from \lref[Theorem]{thm:sumcheck}. This will serve as our outermost
dPCP. Let us recall the parameters of this dPCP. Observe that for this
setting $m_0 = \log_{h_0} n_0 = \lg N/\lg h_0= 10(\lg\lg N)^{20/9}$. $\D^{(0)}$
is a 2-prover decodable PCP with respect to the encoding
$\LDE_{\F,h_0}$ for the language $\ACSAT_\F$ with the following
parameters: On inputs $\Phi$ of size $N$ over $\F$, $\D^{(0)}$ has
randomness complexity $R_0= c\cdot m_0\lg |\F| = 10 c \lg N$, answer size $s_0 = 2(m_0 h_0)^2 < 2^{0.3 (\lg N)^{1-\epsilon}} =
N^{0.3/(\lg\lg N)^{20/9}}$ and distributional soundness error $1/|\F|^{0.1}$.

Let $\epsilon' = \epsilon/10 = 2\lg\lg\lg N/9\lg\lg N$. Let $i^*$ be the
smallest integer such that $1-\epsilon-i\epsilon' < 9\epsilon/80$. Note that
$i^* = O(1/\epsilon) =O(\lg\lg N/\lg\lg\lg N)$. For $i = 1, \dots, i^*$,
let $\D_i$ be the dPCP obtained by instantiating the dPCP in
\lref[Theorem]{thm:sumcheck} with parameters $h_i = 2^{(\lg
  N)^{1-\epsilon-i\epsilon'}} = N^{1/(\lg \lg N)^{20/9\cdot (1+i/10)}}$ and $l_i=2i$. We will run dPCP $\D_i$ on
inputs of instance size $n_i = 2^{3(\lg
  N)^{1-\epsilon-(i-1)\epsilon'}} = N^{3/(\lg \lg N)^{20/9\cdot (1+(i-1)/10)}}$. Thus, $m_i = \lg n_i/\lg h_i = 3
(\lg N)^{\epsilon'} = 3 (\lg\lg N)^{2/9}$. Hence, $\D_i$ is a
$(2i+1)$-answer 2-prover dPCP that on
inputs of instance size $n_i$ has randomness complexity $R_i = c m_i \lg |\F| = 3c
(\lg N)^{1-\epsilon+\epsilon'} = 3c\lg N/(\lg \lg N)^{2}$, answer size $s_i = 2(m_i h_i)^2
< 2^{3 (\lg N)^{1-\epsilon-i\epsilon'}} = N^{3/(\lg \lg N)^{20/9\cdot (1+i/10)}}$ and distributional soundness
error $\delta_i=1/|\F|^{0.1}$.

  Observe that our setting of parameters satisfy $s_{i-1} \leq n_i$ and
  $l_{i+1} +1=2(i+1)+1= (l_i +1)+2$. So the answer size of the predicates produced by dPCP
  $\D_i$ are valid input instances for dPCP $\D_{i+1}$, for
  $i=0,\dots,i^*-1$. Hence, we can compose them with each
  other. Consider the dPCPs $\D^{(i)}$ defined as follows:
\begin{align*}
\D^{(i)} := \D^{(i-1)} \twocomp \D_i, \qquad i = 1,\dots, i^*.
\end{align*}

Also note that the code $\LDE_{\F,h_i}$ has block length $|\F|^{m_i}$ and
distance $(1-O(m_ih_i)/|\F|) \geq 1- 1/\sqrt{|\F|}$. Thus, the
agreement parameter $\eta_i$ is at least $1/|\F|^{1/6}$.

Let $\D^{(I)}:= \D^{(i^*)}$ be the final dPCP obtained as above. Observe
that it is a $2i^* =O(1/\epsilon) = O(\lg\lg N/\lg\lg\lg N)$-prover dPCP with respect
to the encoding $\LDE_{\F,h_0}$ for the language $\ACSAT_\F$ with the
following parameters: On inputs $\Phi$ of size $N$ over $\F$,
the dPCP $\D^{(I)}$ has randomness complexity $R^{(I)}$, distributional soundness
error $\delta^{(I)}$ and
answer size $s^{(I)}$ (which are calculated below).
\begin{align*}
R^{(I)} &= R_0 + \sum_{i=1}^{i^*} \left(R_i  + \log
(\text{blocklength}(\LDE_{\F,h_i} )) \right)\\
& = 10 c\lg N + \sum_{i=1}^{i^*} \left(c m_i \lg |\F| +
m_i \lg |\F|\right)\\
& = 10 c \lg N + \sum_{i=1}^{i^*} 3(c+1) \frac{\lg N}{(\lg \lg N)^{2}}\\
& = 11 c \lg N \qquad [\text{ since } i^* \leq \lg\lg N ].\\
s^{(I)} &= s_{i^*}\\
& = 2^{3 (\lg N)^{1-\epsilon-i^*\epsilon'}}\\
& \leq 2^{3 (\lg N)^{9\epsilon/80}}\\
& = 2^{3(\lg\lg N)^{1/4}}.\\
\delta^{(I)} &= \delta_0 + \sum_{i=1}^{i^*} (\delta_i + \eta_i)\\
& = (i^*+1)\cdot \left(\frac{1}{|\F|^{0.1}} + \frac{1}{|\F|^{1/6}}\right)\\
& \leq \frac{1}{|\F|^{0.05}}.
\end{align*}

\noindent {\sf Stage II:} We now compose the dPCP $\D^{(I)}$ constructed in
{\sf Stage I} with another dPCP obtained from
\lref[Theorem]{thm:sumcheck} as follows. Let $\D_{II}$ be the dPCP
obtained from \lref[Theorem]{thm:sumcheck} by setting $h=2$ and
$l=2i^*$. This dPCP will run on inputs of instance size $n_{II}\geq s^{(I)} = 2^{3
  (\lg N)^{9\epsilon/80}}=2^{3(\lg\lg N)^{1/4}}$. Thus, $m_{II} = \lg n_{II}/\lg h = 3(\lg
N)^{9\epsilon/80} = 3 (\lg\lg N)^{1/4}$. Thus, $\D_{II}$ is a 2-query $(2i^*+1)$-answer dPCP  on inputs of instance size $n_{II}$, randomness $R_{II} = c
\cdot m_{II} \log |\F| = 3c (\lg N)^{1-71\epsilon/80} = 3c \lg N/(\lg\lg N)^{71/36}$, answer size $s_{II} =
2(m_{II} h_{II})^2 < O( (\lg N)^{9\epsilon/40}) =O((\lg \lg N)^{1/2})$ and distributional
  soundness error $\delta_{II}=1/|\F|^{0.1}$. Let $\D^{(II)}$ be the dPCP obtained
  by composing dPCP $\D^{(I)}$ obtained in the previous stage with dPCP
  $\D_{II}$, i.e., $\D^{(II)} = \D^{(I)} \twocomp \D_{II}$. The encoding $\LDE_{\F,h}$ has blocklength $|\F|^{m_{II}}$ and
    distance $1-O(m_{II}h)/|\F| \geq 1- 2/|\F|$. Hence, its agreement
    parameter is at least $1/|\F|^{1/6}$. Thus, dPCP $\D^{(II)}$ is a $2(i^*+1)$-prover dPCP with respect
  to the encoding $\LDE_{\F,h_0}$ for the language $\ACSAT_\F$ with
  the following parameters: On inputs $\Phi$ of size $N$ over $\F$,
the dPCP $\D^{(II)}$ has randomness complexity $R^{(II)} = R^{(I)} + R_{II} +
m_{II}\lg |\F|
= O(\lg N)$, distributional soundness
error $\delta^{(II)} = \delta^{(I)} + \delta_{II} +\eta_{II} \leq \frac{1}{|\F|^{0.05}}$ and
answer size $s^{(II)} = s_{II}= O(\sqrt{\lg\lg N})$. \\

\noindent {\sf Stage III:} We now compose dPCP $\D^{(II)}$ with the
Hadamard based dPCP constructed in \lref[Theorem]{thm:hadamard} to
obtain our final dPCP. Let $\D_{III}$ be the Hadamard based dPCP
constructed in \lref[Theorem]{thm:hadamard} with $l=2(i^*+1)$, i.e.,
$\D_{III} = \D_{\QF,\F,2(i^*+1)}$. $\D_{III}$ will be run on instances
of size $n_{III} = O(\sqrt{\lg\lg N})$. Thus, $\D_{III}$ is a
2-prover $(2i^*+3)$-answer dPCP with respect to the encoding $\QF_\F$
for the language $\ACSAT$ with the following parameters: on inputs of
instance size $n_{III}$, it has
randomness complexity $R_{III} = O(n^2_{III} \lg |\F|) = O(\lg N)$, answer size $s_{III} = O(i^*)$ and distributional soundness error
$\delta_{III} = 1/|\F|^{0.1}$. Furthermore, the blocklength of the
encoding is $|\F|^{O(n_{III}^2)}$ and has agreement parameter
$1/\sqrt{|\F|}$. The final dPCP $\D^{(III)}$ is
given by composing $\D^{(II)}$ with $\D_{III}$, i.e., $\D^{(III)} =
\D^{(II)} \twocomp \D_{III}$. Note that $s^{(II)} \leq
n_{III}$.
Thus, the final dPCP $\D^{(III)}$ is a $2(i^*+2)$-prover dPCP with respect
  to the encoding $\LDE_{\F,h_0}$ for the language $\ACSAT_\F$ with
  the following parameters: On inputs $\Phi$ of size $N$ over $\F$,
the dPCP $\D^{(III)}$ has randomness complexity $R^{(III)} = R^{(II)} + R_{III}
+ O(n_{III}^2 \lg |\F|)
= O(\lg N)$, distributional soundness
error $\delta^{(III)} = \delta^{(II)} + \delta_{III} +\eta_{III} \leq  \frac{1}{|\F|^{0.05}}$ and
answer size $s^{(III)} = s_{III}= O(i^*) = O(1/\epsilon)$.

Summarizing, we have constructed a $O(\lg\lg n/\lg\lg\lg n)$-prover
dPCP $\D^{(III}) for \ACSAT_\F$ with respect to the encoding
$\LDE_{\F,h_0}$ with the following parameters: on inputs $\Phi$ of
size $N$,  $\D^{(III)}$ has randomness complexity $O(\lg N)$, answer
size $O(\lg\lg N/\lg\lg\lg N)$ and and distributional soundness error
$N^{1/\poly\lg\lg N}$. This dPCP implies a PCP for \CSAT\ with
parameters as stated in \lref[Theorem]{thm:maincons}. Note that the
answer size is larger by a factor of $\log |\F| = \lg N/\poly\lg\lg N$
since the size of the output predicate is measured in terms of its
Boolean circuit complexity as opposed to arithmetic complexity.
\qed

\subsection{Optimality of our parameter choices}\label{sec:optimal}

In this section, we show the optimality of the parameters (upto
constants) obtained in our \lref[Theorem]{thm:maincons} using the
Reed-Muller based dPCP (\lref[Theorem]{thm:sumcheck}) and the Hadamard based
(\lref[Theorem]{thm:hadamard}) dPCP as building blocks in our
composition paradigm. Of course, if one had an improved building
block, then one can potentially improve on the construction.

Let $N$ be the size of the instance and let $\delta$ be the
soundness error of the construction. Define parameter $\epsilon$ as
follows: $\log(1/\delta) = (\log N)^{1-\epsilon}$. Consider any
sequence of compositions of the Reed-Muller based dPCP and Hadamard
based dPCP. Observe that the size of the Hadamard based dPCP is
exponential in its input instance size. Hence, the first sequence of
composition steps must involve only the Reed-Muller based dPCP wherein
the size of the instance is sufficiently reduced to allow for
composition with the Hadamard based dPCP.

We first argue that one needs to perform at least $\Omega(1/\epsilon)$ steps
of composition of the Reed-Muller based dPCP so that the instance size
is sufficiently small to apply the Hadamard based dPCP.  Suppose we
perform $t$ steps of composition of the Reed-Muller based dPCP wherein
at the $i$-th step the instance size drops from $N_{i-1}$ to $N_i$
(here, $N_0=N$). Since the
error at each step is at most $\delta$, the field size used in each stage of
the Reed-Muller based dPCP must be at least $1/\delta = 2^{{\log
    N}^{1-\epsilon}}$. To maintain polynomial size of the overall
construction, each of the Reed-Muller based dPCPs
used in the $t$ steps of composition must satisfy $|\F_i|^{m_i} =
N^{O(1)}$ where $F_i$ and $m_i$ are the field and dimension used in
the construction of the Reed-Muller based dPCP used in the $i$-th stage
of the composition. Hence, $m_i \leq O( (\log N)^\epsilon)$. Thus,
the reduction in size in the $i$-th step is at most $N_i \geq N_{i-1}^{1/m_i}=
N_{i-1}^{1/(\log N)^\epsilon}$, which implies inductively that the instance
  size after $t$ steps of composition of the Reed-Muller based dPCP is
  at least $2^{{\log N}^{1-t\epsilon}}$. Hence, to obtain a size that allows
    for composition with the Hadamard-based dPCP we must have at least
    $t = \Omega(1/\epsilon)$ steps of composition.

We now account for the total randomness used in these $t =
\Omega(1/\epsilon)$ steps of composition. Since the error in each step
is at most $\delta$, the randomness uses in each step must be at least
$\log(1/\delta) = (\log N)^{1-\epsilon}$. Hence, the total randomness
used in these $t$ steps is at least $t\cdot \log (1/\delta) =
\Omega(1/\epsilon) \cdot (\log N)^{1-\epsilon}$. Since the size of the
entire construction is at most polynomial we must have that
$1/\epsilon \cdot (\log N)^{1-\epsilon}  = O(\log N)$. Solving for
$\epsilon$\footnote{$1/\epsilon \cdot (\log N)^{1-\epsilon}  = O(\log
  N)$ implies that $1/\epsilon \leq O((\log N)^{\epsilon}$ or
  equivalently $1/\epsilon \cdot \log (1/\epsilon) \leq O(\log
\log N)$. This implies that $\epsilon \geq \log \log \log N/\log \log
N$.}, we obtain that $\epsilon \geq \log\log \log N/\log \log
N$. Hence, the best soundness error obtained by a sequence of
composition involving the Reed-Muller and Hadamard based dPCPs is at
least $\delta = 2^{-(\log N)^{1-\epsilon}} = N^{1/\poly \log \log N}$
proving optimality of the \lref[Theorem]{thm:maincons} construction.

\section{Construction of specific dPCPs}\label{sec:components}

In this section, we construct our two building blocks; the
Hadamard-based dPCP (\lref[Theorem]{thm:hadamard}) and the
Reed-Muller-based dPCP (\lref[Theorem]{thm:sumcheck}). Our construction
proceeds by adapting previous constructions of these objects which
guaranteed only list-decoding soundness. We obtain distributional
soundness by observing that if the dPCP satisfies list-decoding
soundness and the encoding has very good distance (nearly 1), then the
dPCP satisfies distributional soundness.

\subsection{Preliminaries}

Let $\F$ be a finite field.

\begin{definition}[Hadamard]
The Hadamard encoding of a string $a\in \F^m$ is a function $h:{\F^m}\to\F$ defined by
\[ \forall \alpha \in \F^m,\quad h(\alpha) = \sum_i \alpha_i a_i.
\]
\end{definition}

\begin{definition}[Quadratic Hadamard]
The Quadratic Hadamard encoding (\QF\ encoding for short) of a string $a\in \F^m$, denoted $QH_a$, is defined to be the Hadamard encoding of the string $w = a \circ b \in \F^{m+{m}^{2}}$ where $b \in \F^{m ^ 2}$ is defined by $b_{im+j}= a_i a_j$ for all $1\le i,j,\le m$ (i.e. $b= a\otimes a$).
\end{definition}
Let $\e_i\in \F^{m+m^2}$ be the unit vector with $1$ on the $i$th coordinate and zeros elsewhere. Observe that if $h=QH_a$ is the quadratic functions encoding of $a$, then for each $1\le i,j\le m$,
\[ h(\e_i) = a_i \qquad\ve \qquad h(\e_{i\cdot m + j}) = a_ia_j  .
\]
Let $H\subset\F$ and denote $h = \card H$.
Fix an arbitrary 1-1 mapping $H\leftrightarrow [h]:=\set{0,1,\ldots,h-1}$. We refer to elements in $H$ as integers in $[h]$ relying on this mapping. For any $m>0$ we map $x=(x_1,\ldots,x_m)\in H^m$ to $\tilde x = x_1 + x_2h + \ldots + x_{m} h^{m-1} + 1 \in [h^m]$.
\begin{definition}[Low Degree Extension]\label{def:LDE}
Given a string $a\in \F^n$, we define its {\em Low Degree Extension} with respect to $H\subseteq \F$, denoted $LDE_a$, as follows. Let $m$ be the smallest integer such that $h^m \ge n$. Let $f:\F^m\to \F$ be the unique function whose degree in each variable is at most $h$, defined on $H^m$ by %
\[\forall x \in H^m,\quad f(x) =
\left\{
  \begin{array}{ll}
    a_{\tilde x} & \tilde x\in [n]; \\
    0 & n< \tilde x \le  h^m .
  \end{array}
\right.
 \]%
and extend $f$ to $\F^m$ by interpolation, and set $LDE_a=f$.
\end{definition}

\begin{claim}\label{claim:two-LDEs}
Let $a\in \F^{h^{m_1}}$, and let $b \in \F^{h^{m_2}-h^{m_1}}$, so that $a\circ b\in \F^{h^{m_2}}$.
If $g_1 = LDE_a$ and $g_2 = LDE_{a\circ b}$ then
\[\forall x_1,\ldots,x_m \in \F^m,\quad g_1(x_1,\ldots,x_{m_1}) = g_2(x_1,\ldots,x_{m_1},\bar 0) .\]
\end{claim}
\begin{proof}
For each $(x_1,\ldots,x_m)\in H^m$ we have
\[g_1(x_1,\ldots,x_m) = a_{\tilde x} = (a\circ b)_{\tilde x} =
g_2(x_1,\ldots,x_m,\bar 0).\]
Thus, $g_1$ and $g_2$ coincide for all points in $H^m$.
As a function of $x_1,\ldots,x_m$, $g_1$ and $g_2(x_1,\ldots,x_m,\bar
0)$ have degree at most $h$ in each variable, so they must coincide
for all points in $\F^m$ too.
\end{proof}

\begin{definition}[Curve]\label{def:curve}
Given $k < |\F|$ and a sequence of $k+1$ points $\tau = (z_0,\ldots,
z_k)$ in $\F^m$, define
\[
curve_{\tau}:\F\to\F^m
\]
to be the polynomial function of degree at most $k$ which satisfies $curve_{\tau}(i) = z_i$ for $i=0,\ldots,k$.\\
\end{definition}

\begin{definition}[Manifold]\label{def:manif}
Given $\tau = (z_1,\ldots, z_k)\in \F^m$, and three points $x_1,x_2,x_3\in \F^m$ define $\gamma_{z_1,\dots,z_k;x_1,x_2,x_3}:\F^4\to\F^m$ to be the following degree $k+1$ function
\[
\gamma_{z_1,\dots,z_k;x_1,x_2,x_3}(t_0,t_1,t_2,t_3) = t_0 \cdot curve_{x_1,z_1,\dots,z_k}(t_1) + t_2 x_2 + t_3x_3.
\]
\end{definition}
Observe that $\gamma_{z_1,\dots,z_k;x_1,x_2,x_3}$ contains the points $z_1,\ldots,z_k$ and $x_1,x_2,x_3$.

We now state a low degree test, which has appeared in several places in the literature \cite{RazS1997,DinurFKRS2011,MoshkovitzR2010b}. First, a little notation. Supposed that $Q:\F^m\to\F$ is a function of degree $\le d$, and $\gamma_{z_1,\dots,z_k,x_1,x_2,x_3}(t_0,t_1,t_2,t_3) = t_0 \cdot curve_{x_1,z_1,\dots,z_k}(t_1) + t_2 x_2 + t_3x_3$ is a manifold in $\F^m$ of degree at most $k+1$. Then the function $Q\circ \gamma:F^4\to\F$ has degree at most $d(k+1)$ and can be specified by ${d(k+1)\choose 4}$ coefficients. Given a manifold $\gamma$ and a function $M(\gamma):\F^4\to\F$, we denote for each $x \in Im(\gamma)$
\[M(\gamma)[x] := M(\gamma)(t_1,\ldots,t_4)\quad \hbox{for }t_1,\ldots,t_4\hbox{ such that }\gamma(t_1,\ldots,t_4)=x.\] %
The following lemma
appears in~\cite[Lemma~4.4, Section~10.2 (in
appendix)]{MoshkovitzR2010b} and a similar lemma can be found in \cite[Lecture~9]{Harsha2010}.
\begin{lemma}[Low Degree Test - Manifold
  vs. Point]\label{lem:LDT}\phnote{Check Citation??}
Let $m,k,d>0$, let $\delta = (mkd/\card\F)^{1/8}$, and let $z_1,\ldots,z_k\in \F^m$ be fixed. Let $Q:\F^m\to\F$ be an arbitrary function, supposedly of degree $\le d$. There exists a list of $L \le 2/\delta$ degree $d$ functions $Q_1,\ldots,Q_L:\F^m\to\F$ such that the following holds.
Let $\Gamma$ be a collection of manifolds,
\[ \Gamma = \set{ \gamma_{z_1,\dots,z_k,x_1,x_2,x_3}(t_0,t_1,t_2,t_3) = t_0 \cdot curve_{x_1,z_1,\dots,z_k}(t_1) + t_2 x_2 + t_3x_3 }_{x_1,x_2,x_3}
\]
one per choice of $x_1,x_2,x_3\in \F^m$.
Let $M:\Gamma\to \F^{d (k+1)\choose 4}$ specify for each $\gamma$ the coefficients of a degree-$d(k+1)$ function supposedly equal to $Q\circ \gamma:\F^4\to\F$.
Then,

\[ \Pr_{z,\gamma\ni z} [ Q(z) = M(\gamma)[z]  \ve Q(z) \not \in \set{Q_1(z),\ldots,Q_L(z)} ] \le \delta.
\]
\end{lemma}
Finally, we state the following lemma, which gives a probabilistic verifier that inputs a predicate $\Phi:\F^n\to\bits$ and a list of functions $F_1,\ldots,F_\ell:\F^n\to\F$, and checks that $(a,b)$ are such that $\Phi(a)=1$ and $b_i = F_i(a)$ for each $i=1,\ldots,{\ell}$ ($b=F(a)$ for short).

\begin{lemma}[Initial Verifier]\label{lem:had-easy-ver}
Given a predicate $\Phi:\F^{{n}}\to\bits$ and functions $F_1,\ldots,F_{{\ell}}:\F^{{n}}\to\F$ whose total circuit complexity is $N$, there is a randomized verifier $V_0$ that uses $O(\log\card\F+\log N)$ random bits and generates a quadratic polynomial $p:\F^m\to\F$ on $m = O(N)$ variables such that, given access to a proof $\pi = a\circ b\circ s \in \F^m$,
\begin{itemize}
\item If $\Phi(a)=1$ and $F_i(a)=b_i$ for each $i=1,\ldots,\ell$, then there is a {\em unique} string $s = s(a,b)$ such that \[\Pr_{p\sim V_0} [p(a,b,s) = 0]=1.\]
\item  If either $\Phi(a)=0$ or $b_i\neq F_i(a)$ for some $1\le i\le \ell$, or $s\neq s(a,b)$, then
    \[\Pr_{p\sim V_0} [p(a,b,s) = 0]\le \frac 2{\card \F}.\]
\end{itemize}
\end{lemma}
This verifier would be ideal except for one drawback: in order to evaluate $p(\pi)$ it makes an {\em unbounded} number of queries to the proof $\pi$.

\begin{proof}(sketch)~
The proof of this lemma is standard: $s$ will specify the values of all of the intermediate gates of the circuit computing $\Phi$ as well as the circuits computing $F_1,\ldots,F_\ell$. The validity of each intermediate computation step can be checked by a quadratic or linear equation over the entries in $s$. The verifier $V_0$ will use its randomness to generate a (pseudo)random sum of these equations (using an error correcting code, details are omitted). This can be expressed as a quadratic polynomial over
the set of new variables.
\end{proof}
%

%
%
%
%

\subsection{Hadamard based dPCP}
In this section we construct a dPCP based on the Hadamard encoding, given formally in the following lemma.

\ \\
\noindent {\bf \lref[Theorem]{thm:hadamard} (Restated)} (Hadamard based dPCP) {\em For any finite field $\F$, and
  any $\ell>0$, there is a 2-prover $\ell+1$-answer decodable PCP $\D_{\QF,\F}$
  with respect to the encoding $\QF_\F$ for the language $\ACSAT_\F$
  with the following parameters: On inputs (i) a predicate
  $\Phi:\F^{n}\to\bits$ and (ii) functions
  $F_1,\ldots,F_{\ell}:\F^{n}\to\F$ given by arithmetic circuits over $\F$ whose total size is $N$, the dPCP $\D_{\QF,\F}$ has
\begin{itemize}
\item randomness complexity $O(N^2 \log |\F|)$,
\item answer size $s,s'=O(\ell)$,
\item perfect completeness, and
\item distributional soundness error $\le 1/|\F|^{0.1}$.
\end{itemize}}\ \\

We define the verifier for \lref[Theorem]{thm:hadamard}.

 \paragraph{Decoder Protocol}
On input $\Phi,F_1,\ldots,F_\ell; j,r$, let $V_0$ be the verifier from \lref[Lemma]{lem:had-easy-ver}, and let $\pi\in\F^m$ be the proof that $V_0$ expects. Our decoder $V$ expects the $B$ prover to hold the \QF\ encoding
of $\pi$ and the $A$ prover is expected to give restrictions of $B$ to specified subspaces. It is known that with $O(1)$ queries into $B$ the decoder could check that $B$ is indeed a \QF\ encoding of a valid proof $\pi$, as well as decode any quadratic function of $\pi$. The $A$ prover is used to simulate this while making only one query to $A$ and one to $B$. This is done by computing several query points for the former test, and then taking a random subspace $S$ containing these points as well as a couple of uniformly random ones.

The low degree test (see \lref[Lemma]{lem:Had-lintest} below) guarantees that if $A$'s answer on the subspace $S$ is consistent with $B$'s answer on a random point in $S$, then $B$ is linear (in other words, it is a Hadamard encoding of some string). The decoder will also perform some other tests on values in $S$ which ensure that moreover $B$ is a valid \QF\ encoding of a valid $\pi$.

\begin{enumerate}
\item Computing the query points.
\begin{enumerate}
\item Choose $\beta,\gamma\in \F^m$ uniformly at random, and define
  $u_1,u_2,u_3 \in \F^{m+m^2}$ as follows: \[u_1 = \sum_{i=1}^m \beta_i \e_{i},\quad u_2 = \sum_{i=1}^m \gamma_i\e_{i}, \quad u_3 = \sum_{i=1}^m\sum_{i'=1}^m \beta_i \gamma_{i'}\e_{im+i'}.\] (These are points for the multiplication test: if we already know that $B$ is a Hadamard encoding of some string, then this test will ensure it is moreover a \QF\ encoding. )
\item Draw a random quadratic polynomial $p(t_1,\ldots,t_m) = \alpha_0 + \sum_i \alpha_it_i + \sum_{i,i'}\alpha_{ii'}t_{i}t_{i'}$ from the distribution of $V_0$ (from \lref[Lemma]{lem:had-easy-ver}). To check that $p(\pi)=0$ we define $z\in \F^{m_2}$ (for $m_2=m+m^2$) by
   \[ z = \sum_{i=1}^m \alpha_i \e_i + \sum_{i=1}^m\sum_{i'=1}^m \alpha_{ii'}\e_{im+i'}\]
   (If $B$ were equal to $QH_\pi$ for some string $\pi$ then $B(z) + \alpha_0 = p(\pi)$ so the value of $B(z)$ could be used to check that $p(\pi)=0$).
\item For each $i = 1,\ldots,{\ell}$ let $o_i = \e_{i+{n}}$. Let
  $o_{\ell+1}$ be the point in $\F^{m_2}$ corresponding to $j=
  (\delta_1,\ldots,\delta_{{n}+{n}^2})\in \F^{{n}+{n}^2}$\phnote{Why
    is $j$ in $\F^{n+n^2}$ and not $\F^{m+m^2}$} \inote{because the dPCP needs to return a value in the QF encoding of the input, whose length is $n$, and not of the proof $\pi$.}. (We are using here the fact that $\pi=a\circ b\circ s$ so the $QF$ encoding of $\pi$ contains in it the QF encoding of $a$. Explicitly,
  set \[o_{\ell+1}  = \sum_{i=1}^{{n}} \delta_i \e_i +
  \sum_{i=1}^{{n}} \sum_{i'=1}^{{n}} \delta_{i n + i'} \e_{im +
    i'}.\]\phnote{What is $n_1$?}\inote{changed $n_1$ to $n$}
    (These are the points to be output by the decoder.)
\end{enumerate}
\item Choose $x_1,x_2,x_3\in\F^{m_2}$ uniformly at random, consider
  the $(l+8)$-dimensional linear subspace
\[S = \spn(x_1,x_2,x_3,\;u_1,u_2,u_3,\;z,\;o_1,\ldots,o_{\ell+1})\subset \F^{m_2}.\]
We assume there is a canonical mapping that maps each subspace $S$ to a particular basis $\vec {v_S}=\set{v_1,\ldots,v_{\ell+8}}\subset\F^{m_2}$ for $S$ and send it to the $A$ prover and let $A(\vec {v_S}) \in \F^{\ell+8}$ be the prover's answer. The answer specifies a linear function $\as : S \to \F$ defined by
\[\forall t_1,\ldots,t_{\ell+8}\in \F,\qquad \as\big(\sum_i t_i v_i\big) := \sum_{i=1}^{\ell+8} t_i \cdot A(\vec {v_S})_i \]
\item Send $x_1$ to the $B$ prover and let $B(x_1)$ be its answer.
\item Reject unless \begin{enumerate}
\item\label{itm:cons} $\as(x_1) = B(x_1)$, and
\item\label{itm:t3} $\as(z) + \alpha_0=0$.
\item\label{itm:t2} $\as(u_1)\as(u_2)=\as(u_3)$, and
\end{enumerate}
\item Output $\as(o_1),\ldots,\as(o_\ell)$.
\end{enumerate}
The decoding PCP will follow the protocol above, using its randomness $R$ for selecting $p,\beta,\gamma,x_1,x_2,x_3$, and generate an output $(q,\varphi,f,g)$ as follows:
\begin{itemize}
\item The queries $q$ are $q_0=\vec {v_S}$ to the first prover and $q_1=x_1$ to the second prover.
\item The predicate $\varphi$ - rejects iff at least one of the tests in Items~\ref{itm:t3} and \ref{itm:t2} reject.
\item The function $g$ computes $\as(x_1)$ (for the consistency test in \lref[Item]{itm:cons}).
\item The functions $f_1,\ldots,f_{\ell+1}$ - compute $\as(o_i)$ for $i=1,\ldots,\ell+1$.
\end{itemize}

\begin{lemma}[Perfect Completeness]
The verifier has perfect completeness. Namely, for every $a\in \Phi^{-1}(1)$, there is a proof $\Pi$ such that for every $j\in \F^{{n}+({n})^2}$ and every random string $R\in \bits^{O(N^2\log\card\F)}$, the verifier on input $(\Phi,F;j,R)$ accepts and outputs $F_1(a),\ldots,F_{\ell}(a),QH_a(j)$.
\end{lemma}
\begin{proof}
Let $b=F(a)$ and let $s$ be the string promised in \lref[Lemma]{lem:had-easy-ver}. Let $B:\F^{m_2}\to\F$ be the quadratic functions encodings of $\pi = a\circ b\circ s$. For each $\vec {v_S}=  (v_i)_i$, let $A(\vec {v_S}) = (B(v_1),B(v_2),\ldots, B(v_{\ell+8}))$. We claim that $\Pi=(A,B)$ is a valid proof for $a\in \Phi^{-1}(1)$:

By definition $B$ is a linear function on $\F^{m_2}$, so $\as(x) = B(x)$ for all $x\in S$ and in particular the test in \lref[Item]{itm:cons} passes. Also, by definition $B$ is the Hadamard encoding of the string $\sigma = \pi \circ (\pi\otimes \pi)$, so $B(\e_i) = \sigma_i$ for all $1\le i \le m+m^2$. Thus
\[ B(\e_{i_1m+i_2}) = \sigma_{i_1m+i_2} = \pi_{i_1}\cdot \pi_{i_2} = \sigma_{i_1}\cdot\sigma_{i_2} = B(\e_{i_1})\cdot B(\e_{i_2}) \]
which, by linearity, implies that the test in \lref[Item]{itm:t2} passes. Next, for \lref[Item]{itm:t3}, we know that for every $p$ generated by $V_0$,
\[ 0= p(\pi)=\alpha_0 + \sum_{i=1}^{{n}} \alpha_i \pi_i + \sum_{i,i'=1}^{{n}}\alpha_{ii'}\pi_i\pi_{i'}= \alpha_0 + \sum_i \alpha_i B(\e_i)  + \sum_{ii'}\alpha_{ii'}B(\e_{im+i'}) = \alpha_0 + B(z),\]
so $\as(z) + \alpha_0 = B(z)+\alpha_0 = 0$ as required. It is finally easy to check that
\[\as(\e_{i+n}) = B(\e_{i+{n}}) = \pi_{i+{n}} = b_i,\qquad
i=1,\ldots,\ell.\] Finally, for the $\ell+1$st output,\inote{changed $n'$ to $n$}
\[\as(o_{\ell+1}) = B(o_{\ell+1}) = \sum_{i=1}^{n} \delta_i B(\e_i) + \sum_{i=1}^{n} \sum_{i'=1}^{n} \delta_{i n +i'} B(\e_{im+i'})  = QH_a(j)\] where the last equality is due to the fact that the \QF\ encoding of $\pi=a\circ b\circ s$ contains the \QF\ encoding of $a$. More precisely, $QH_a(\e_{i n +i'}) = B(\e_{im+i'})$ for all $0\le i\le n$ and $1\le i' \le n$.
\end{proof}

\begin{lemma}[Distributional Soundness]\label{lem:Had-soundness}
The verifier above has soundness error at most $\delta = \card \F^{-0.1}$. Namely, given $(\Phi,F)$ for every proof $\Pi=(A,B)$, there are functions $\fakePi(\cdot),\fakex(\cdot)$ such that
\begin{itemize}
  \item For each $R$, either $\Phi(\fakex(R))=1$ and $\fakePi(R)$ is a valid proof for ``$x\in SAT(\Phi)$'' or $\fakePi(R)=\bot$.
  \item For every $j$, there is probability at least $1-\epsilon$ that when
  $R$ is chosen randomly and $V$ is run on $(\Phi,F;j,R)$ it
  either rejects, or $\fakePi(R)$ is a proof that completely agrees with
  the answers of the provers $A,B$ on the queries of $V$ (in
  which case $V$'s output is consistent with $\fakex(R)$).
\end{itemize}
\end{lemma}

\begin{proof}
Fix $\Pi = (A,B)$. Given $B$, let $g_1,\ldots,g_L$ be as in the low degree test below, \lref[Lemma]{lem:Had-lintest}. Let
\[ L^* = \sett{ i\in [L]}{g_i=QH_\pi\hbox{ for }\pi=a\circ b\circ s\hbox{  s.t.  }\Pr[V_0^\pi\hbox{ accepts}]=1}.\]
 Set $\fakePi(R)=\bot$ if events $E1$ or $E2$ occurred, where
\begin{enumerate}
  \item[E1:] $B(x_1) \not \in \set{g_i(x_1)\;|\;i\in L^*}$.
  \item[E2:] there is more than one index $i\in L^*$ for which $B(x_1)=g_i(x_1)$.
\end{enumerate}
Otherwise, there is a unique $i\in L^*$ such that $B(x_1)=g_i(x_1)$. By assumption $g_i$ is the \QF\ encoding of some $\pi = a\circ b \circ s$ for which $\Phi(a)=1$ and $F(a)=b$ and $s=s(a,b)$. So we set $\fakex(R)=a$ and set $\fakePi(R) = (A_R,B_R)$ to be a valid proof for $a\in\Phi^{-1}(1)$ so that $B_R = g_i$.

Now fix an arbitrary $j\in \F^{n+n^2}$, and let $R$ be chosen uniformly at random. We claim that the probability that the verifier accepts and yet the view of $\Pi$ and of $\fakePi(R)$ differ is very small. We analyze two cases.
\begin{itemize}
  \item {\bf Accept and $\fakePi(R)=\bot$:} This event can be bounded by
  \[ \Pr[ \hbox{Accept and }E1] +\Pr [ E2 ]  \le \max\left( 2/\card\F^{1/6}  , \frac {4L}{\card\F } \right)+ \binom{L}2 /\card\F \]
  where the second item is bounded due to the large distance of the Hadamard code, and the first item is bounded as follows. If $B(x_1) \not\in \set{g_i(x_1)\;|\; i\in [L]}$ then \lref[Lemma]{lem:Had-lintest} implies that the probability of acceptance is small. If however $B(x_1) = g_i(x_1)$ for some $i\in [L]\setminus L^*$ then for each $i\in [L]\setminus L^*$ \lref[Lemma]{lem:QF-linear-soundness} shows that the acceptance probability is small, and we take a union bound over all such $i$.
  \item {\bf Accept and $\fakePi(R)|_q \neq \Pi|_q$}: We defined $\fakePi(R)$ so that $B_R(x_1) = B(x_1) =g_i(x_1)$ for some $i\in L^*$. So this event occurs if $A_S \neq g_i|_S$. We observe that this event is contained in $\cup_{i\in L^*} E_i$ where $E_i$ is the event that $A_S \neq g_i|_S$ yet $A_S(x_1) = g_i(x_1)$. For each $i$ this event has probability at most $2/\card\F$, and we take a union bound over $i\in L^*$.
\end{itemize}
\end{proof}

The proof of soundness is based on the following lemma, which has
appeared in several places in the literature. The following lemma
appears in~\cite[Proposition 11.0.3]{MoshkovitzR2010b}.

\begin{lemma}[Subspace vs. Point - linearity testing list decoding soundness]\label{lem:Had-lintest}
Let $\delta=2/\card\F^{1/6}$. Given a pair of provers $A,B$, there is a list of $L \le 2/\delta^3$ linear functions $g_1,\ldots,g_L:\F^{m_2}\to\F$ such that the probability that the decoder does not reject yet $B(x_1) \not\in \set{g_1(x_1),\ldots,g_L(x_1)}$ is at most $O(\delta)$.
\end{lemma}

The following claim shows that if $B$'s answers are a {\em linear} function, then the verifier rejects unless $B$ is a \QF\ encoding of a valid proof.
\begin{lemma}\label{lem:QF-linear-soundness}
  Suppose that $B:\F^{m_2}\to\F$ is a {\em linear function}. Let $\pi = a\circ b\circ s$ be defined by
  \[a = B(\e_1)\ldots B(\e_{{n}})\ve  b = B(\e_{{n}+1}),\ldots,B(\e_{{n}+{\ell}}) \ve  s = B(\e_{{n}+\ell+1}),\ldots,B(\e_{m}).\] Assume that either $\Phi(a)=0$ or $F(a)\neq b$ or $s\neq s(a,b)$ or $B \neq QH_\pi$. Then, for all provers $A$, for all $j$, $\Pr_R[V^{A,B}(\Phi,F;j;R)\hbox{ accepts}]\le 4/\card\F$.
\end{lemma}
\begin{proof}
Assume first that $B = QH_\pi$, but either $\Phi(a)=0$ or $F(a)\neq b$ or $s\neq s(a,b)$. By \lref[Lemma]{lem:had-easy-ver}, when choosing a random quadratic $p$, there is at most probability $2/\card\F$ that $p(\pi) = 0$. So let $p$ be such that $p(\pi)\neq 0$. Since $B$ is the \QF\ encoding of $\pi$, by the definition of $z$ %
\[ B(z)+\alpha_0 = p(\pi) \neq 0.\]
However, $V$ accepts so \lref[Item]{itm:t3} passing implies that $\as(z)+\alpha_0 = 0$. This means that as linear functions on $S$, $\as \not\equiv B|_S$. It remains to observe that conditioned on $S$, $x_1$ is drawn almost uniformly from $S$, so the probability that \lref[Item]{itm:cons} does not reject is at most $1/\card\F + neg \le 2/\card\F$.
Altogether the total probability of accepting in this case is at most $4/\card\F$.

We move to analyze the case where $\Phi(a)=1$ and $F(a)=b$ and $s=s(a,b)$ but $B\neq QH_\pi$. There must be some indices $i_1,i_2$ such that $B(\e_{i_1})B(\e_{i_2})\neq B(\e_{i_1m+i_2})$. We need to upper bound the probability over random $\beta,\gamma$ that the following expression is zero:
\[B(u_1)B(u_2) - B(u_3) = \sum_{i=1}^m \sum_{i'=1}^m  \beta_i \gamma_{i'} ( B(\e_{i})B(\e_{i'}) - B(\e_{im+i'})) .\]
Let us fix the value of $\beta_{i}$ (for $i\neq i_1$) and $\gamma_{i'}$ (for $i'\neq i_2$) arbitrarily. The remaining expression becomes a non-zero quadratic polynomial in $\beta_{i_1},\gamma_{i_2}$, so it can be zero with probability at most $2/\card\F$ over the choice of $\beta_{i_1},\gamma_{i_2}$.

Consider the event where it is not zero. Since \lref[Item]{itm:t2} accepts, $\as(u_1)\as(u_2) = \as(u_3)$ so $\as(u_i) \neq B(u_i)$ for some $i\in \set{1,2,3}$. So as linear functions $\as\neq B|_S$ and the probability of \lref[Item]{itm:cons} passing is at most $1/\card\F + neg \le 2/\card\F$ (where $neg$ is a small probability introduced because $x_1$ is only almost uniform in $S$ conditioned on $S$). Altogether we get a bound of $4/\card\F$ in this case as well.
\end{proof}

\subsection{Reed-Muller based dPCP}
\newcommand{\wX}{{\widehat X}}
\newcommand{\wY}{{\widehat Y}}
\newcommand{\wA}{{\widehat A}}
\newcommand{\wB}{{\widehat B}}
\newcommand{\wpsi}{{\widehat \psi}}
\newcommand{\az}{{\sc all-zero~}}

\inote{Soundness works for params s.t. $h^2 m^3 \card\F^{1/8}\ll  \card \F^{-0.1}$.}
\noindent {\bf \lref[Theorem]{thm:sumcheck} (Restated)} (Reed-Muller based dPCP) {\em For any finite field $\F$,
  and parameter $h$ such that $1 < h < |\F|^{0.01}$ and any $\ell>0$, there
    is a 2-prover $\ell+1$-answer decodable PCP $\D$ with respect
    to the
    encoding $\LDE_{\F,h}$ for the language $\ACSAT_\F$ with the following parameters:
    On inputs (i) a predicate $\Phi:\F^{n}\to\bits$ and (ii) functions
    $F_1,\ldots,F_{\ell}:\F^{n}\to\F$ given by arithmetic circuits over $\F$ whose total size is $N$, the dPCP $\D$ has (let $m=\log N/\log h$)
\begin{itemize}
\item randomness complexity $O(\log {N}+ m\log |\F|) = O(\log {N} + \log \card\F)$,
\item answer size $s,s'=O(m(m+\ell))$,
\item and distributional soundness error $1/|\F|^{0.1}$.
\end{itemize}}\ \\

Fix $H\subseteq \F$ throughout this section and denote $h=\card H$.

We construct a PCP decoder that receives as proof a sequence of low degree polynomials that allow it to simulate the actions of the initial verifier $V_0$ (from \lref[Lemma]{lem:had-easy-ver}) using fewer queries. We first construct this sequence of polynomials $g_1,g_2,g_3,g_4$, then describe a ``verification protocol'' checking that a given sequence of polynomials have the intended form, and finally describe the PCP decoder.

\paragraph{Constructing the low degree functions}
Let $\Phi,F_1,\ldots,F_\ell$ be the input.
\begin{enumerate}
\item Suppose $a\in\Phi^{-1}(1)$ and let $b_i=F_i(a)$ for all $i=1,\ldots,\ell$, and let $s=s(a,b)$ so that the initial verifier $V_0$ accepts $\pi=a\circ b\circ s$ with probability $1$.
    Wlog we assume that $n = \card a$ is a power of $h$, and also $n_1=\card\pi+1$ is a power of $h$. This can be arranged by padding $a$ with zeros and then padding $\pi$ with zeros and changing $V_0$ accordingly.
\item Let $m_1 = m = \log_h n$ and define $g_1 = LDE_a:\F^{m_1}\to\F$ (see \lref[Definition]{def:LDE}),
\item Let $m_2=\log_h n_1$ and let $g_2 = LDE_{\pi\,\circ\,1}:\F^{m_2}\to\F$ be the low degree extension of the string obtained by appending a $1$ to $\pi$. By \lref[Claim]{claim:two-LDEs}
    $g_2(x_1,\ldots,x_{m_1},\bar 0) = g_1(x_1,\ldots,x_{m_1})$.
    Let $z_0\in H^{m_2}$ be the point associated with the last element in $\pi\circ 1$, i.e. such that $g_2(z_0)=1$.
\item\label{itm:g3} Let $\m = 2m_2$ and let $g_3:\F^{\m}\to \F$ be defined by $g_3(x,y) = g_2(x)\cdot g_2(y)$. Note that the degree of $g_3$ is at most $\m{h}$.
\item\label{itm:quad}  Let $P_0$ be the set of all quadratic polynomials generated by $V_0$ on input $(\Phi,F)$. Fix some $p\in P_0$, $p(t_1,\ldots,t_{n_1}) = p_0+\sum_i p_i t_i + \sum_{ii'} p_{ii'} t_it_{i'}$. Define the function $\hat p:\F^\m\to \F$ as follows. For $1\le i \le h^{m_2}$ let $\vec i $ be the corresponding element in $H^{m_2}$ (see discussion before \lref[Definition]{def:LDE}). For each $i<i'\in H^{m_2}\setminus\set {z_0}$, set
\[ \hat p(z_0,z_0) = p_0,\quad \hat p(z_0,\vec i) = p_i,\quad \hat p(\vec i,\vec i') = p_{ii'} ,\quad \hat p(z)=0\hbox{ for all other }z\in H^\m.
\]
Extend $\hat p$ from $H^\m$ to $\F^\m$ by interpolation. The degree of $\hat p$ is at most $\m{h}$.

This definition ensures that for $\sigma := (\pi\circ 1)  \in \F^{n_1}$,
\begin{equation}\label{eq:p-sumsto-zero}
p(\sigma)= p_0+\sum_i p_i \sigma_i + \sum_{ii'} p_{ii'} \sigma_i \sigma_{i'} = \sum_{x\in H^\m} \hat p(x) \cdot g_3(x)= 0.
\end{equation}
\item\label{def:sc-polys} Define low degree functions $s^p_1,\ldots,s^p_\m:\F^\m\to\F$ as nested partial sums of $g_3$ as follows.
\[
 s^p_\m(x) =  \hat p(x)\cdot  g_3(x) \qquad \ve \qquad
 s^p_{i-1}(x) = \sum_{h\in H} s^p_{i}(x_1,\ldots,x_{i-1},h,0,\ldots,0) \quad 1<i\le \m \]
The degree of each $s^p_i$ is at most $2\m h$.
\item Bundling: From the polynomials $g_3$ and $s_i^p$ for each $i$ and $p\in P_0$, we will now create one single polynomial $g_4$ that `bundles' them together. Let us number them as $q_1,\ldots,q_T$ for $T = \m\card{P_0}+1$ and let $t = \ceil{\log_{{h}}(T+1)}=O(\m)$.

    Let $m_4 = \m+t$ and define $g_4:\F^{m_4}\to \F$ by
\begin{equation}\label{eq:bundling} g_4(y_1,\ldots,y_t,x) = \sum_{i=1}^T q_i(x)\cdot w_i(y_1,\ldots,y_t)
\end{equation}
where $w_i:\F^t\to\F$ is the degree $t{h}$ polynomial for which $w_i(y_1,\ldots,y_t) =1$ iff $y_1,\ldots,y_t$ is the $H$-ary representation of $i$, and zero for all other $y_1,\ldots,y_t \in H^t$. The degree of $g_4$ is at most \[d=ht+2h\m =O(hm).\]
\end{enumerate}

By construction the function $g_4$ contains inside of it (as restrictions) the functions $g_1,g_2,g_3$:
\begin{claim}\label{claim:mappings} There is a sequence of 1-1 (linear) mappings
\[\F^{m_1} \stackrel{\sigma_1}{\to} \F^{m_2} \stackrel{\sigma_2}{\to} \F^{\m} \stackrel{\sigma_3}{\to} \F^{m_4} \]
such that for each $x\in \F^{\m}$ we have $g_4(\sigma_3(x)) = g_3(x)$; and
for each $x\in \F^{m_2}$ we have $g_3(\sigma_2(x)) = g_2(x)$; and
for each $x\in \F^{m_1}$ we have $g_2(\sigma_1(x)) = g_1(x)$. \qed
\end{claim}
\begin{proof}
We map a point $x_1 \in \F^{m_1}$ to $(x_1,\bar 0)\in \F^{m_2}$. We map a point $x_2 \in \F^{m_2}$ to $(z_0,x_2)\in \F^\m$. We map a point $x_3\in \F^\m$ to $(y,x_3)\in \F^{m_4}$ in the domain of $g_4$, where $y\in\F^t$ is the index of $g_3$ in the bundling.
\end{proof}
Answers of the $B$ prover will correspond to point-evaluations of $g_4$, and answers of the $A$ prover will correspond to restrictions of $g_4$ to certain low-degree curves. A low degree curve, see \lref[Definition]{def:curve}, is specified by a tuple of points in $\F^{m_4}$ through which it passes. This tuple contains
 \begin{itemize}
 \item Points for the verification check protocol, see below
 \item Output points: The are $\ell+1$ points whose values will give $\ell+1$ answers for the decoder. These answers are the values of $b_1,\ldots,b_\ell$ and the $j$-th element in the encoding $LDE_a$.

    Let $v_1,\ldots,v_\ell\in \F^{m_2}$ be the points in the domain of $g_2$ that correspond to $b_1,\ldots,b_{\ell}$. For each $i$, let $o_i = \sigma_3(\sigma_2(v_i))$ be the corresponding point in the domain of $g_4$ (as in the claim above).

    Let $v_{\ell+1}\in \F^{m_1}$ be the $j$th point in the domain of $g_1=LDE_a$ (we assume some canonical numbering of the indices of the LDE encoding). Let $o_{\ell+1} = \sigma_3(\sigma_2(\sigma_1(v_{\ell+1})))$ be the corresponding point in the domain of $g_4$.

 \item Random points
 \end{itemize}
\paragraph{Verification Protocol}
The verification protocol accesses functions $\tilde g_3$ and $\set{{\tilde  s}_i^p}_{i,p}$ and checks (locally) that they have the correct form, as intended in the construction, i.e. that there is some valid proof $\pi$ such that they are equal to $g_3$ and $s_i^p$ as in the construction above.

  \begin{enumerate}
  \item Check that $\tilde g_3(z_0,z_0)=1$.
  \item Choose two random points $x,y\in \F^{m_2}$, so that $(x,y)\in\F^{\m}$. Check that $\tilde g_3(z_0,x)\cdot \tilde g_3(z_0,y) = \tilde g_3(x,y)$.
  \item Choose a random quadratic $p$ by simulating $V_0$, and compute $\hat p:\F^\m\to\F$ as in \lref[Item]{itm:quad} in the expected proof.
  \item\label{itm:sumcheck-step} Do the sumcheck: Choose a random point $x = (x_1,\ldots,x_\m)\in \F^\m$, and do
  \begin{enumerate}
    \item\label{itm:4a} Check that $\tilde s^p_\m(x) = \hat p(x) \cdot \tilde g_3(x)$
    \item\label{itm:4b} For each $1<i\le \m$ check that $\tilde s^p_{i-1}(x_1,\ldots,x_{i-1},\bar  0) = \sum_{h\in H} \tilde s_i^p(x_1,\ldots,x_{i-1},h,\bar 0)$
    \item\label{itm:4c} Check that $\sum_{h\in H} \tilde s_1^p(h,\bar 0)=0$.
  \end{enumerate}
  \end{enumerate}
The verification protocol accesses $k=(h+1)m_3+5$ points in the
domains of $\tilde g_3$ and $\tilde s_i^p$. Let $u_1,\ldots,u_k\in
\F^{m_4}$ be the corresponding points in the domain of $g_4$ (as in
\lref[Claim]{claim:mappings}).

\paragraph{The PCP Decoder protocol}
\begin{enumerate}
\item Compute the output points $o_1,\ldots,o_{\ell+1}$ as above.
\item Use the randomness $R$ to compute $u_1,\ldots,u_M$ using the verification protocol.
\item Use the randomness $R$ to choose $x_1,x_2,x_3\in \F^{m_4}$ uniformly. Let $\gamma =
\gamma_{o_1,\ldots,o_{\ell+1},u_1,\ldots,u_M;x_1,x_2,x_3}$ be the manifold as in \lref[Definition]{def:manif}, such that $\gamma$ contains the points $o_1,\ldots,o_{\ell+1},u_1,\ldots,u_M$ as well as $x_1,x_2,x_3$ and has degree at most $\ell+1 + k+1 = \ell+O(hm)$. Send $\gamma$ to $A$ and let $A$'s answer $A(\gamma)$ be the coefficients of a function $\F^4\to\F$ whose degree is at most $d' \le d(M+\ell+2)$. This function is supposed to equal $B\circ \gamma$.

Let $\ag:Im(\gamma)\to\F$ be defined for each $x\in Im(\gamma)$ as $\ag(x) := A(\gamma)(t)$ where $x = \gamma(t)$. Clearly $\ag$ can be computed from $A(\gamma)$.
\item Send $x_1$ to the $B$ prover and let $B(x_1)$ be its answer. Reject unless
 $\ag(x_1) = B(x_1)$.
\item Simulate the checks of \lref[Item]{itm:sumcheck-step} using the values $\ag(u_1),\ldots,\ag(u_{M})$. Reject unless all of the checks succeed.
\item Output $\ag(o_1),\ldots,\ag(o_{\ell+1})$.
\end{enumerate}
To summarize, the PCP decoder computes $(q,\varphi,f,g)$ as follows:
\begin{itemize}
\item The queries $q$: $q_0=\gamma$ is the query to the $A$ prover and $q_1=x_1$ is the query to the $B$ prover.
\item The predicate $\varphi$ rejects unless all of the checks in \lref[Item]{itm:sumcheck-step} pass.
\item The function $g$ computes $\ag(x_1)$ (for the consistency test).
\item The functions $f_1,\ldots,f_{\ell+1}$ - compute $\ag(o_i)$ for $i=1,\ldots,\ell+1$.
\end{itemize}
This completes the description of the PCP decoder, and we proceed to prove its correctness.

\begin{lemma}[Perfect Completeness]\label{lem:sc-completeness}
The PCP decoder has perfect completeness. Namely, for every $a\in \Phi^{-1}(1)$, there is a proof $\Pi$ such that for every $j\in \F^{m}$ and every random string $R$, the verifier on input $(\Phi,F;j,R)$ accepts and outputs $F_1(a),\ldots,F_{\ell}(a),LDE_a(j)$.
\end{lemma}
\begin{proof}
   If $\Phi(a)=1$ and $F(a)=b$ then there is a proof $\pi=abs\in \F^{n_1}$ such that for every quadratic $p$ generated by the initial verifier $V_0$, $p(\pi)=0$. Compute from $\pi$ the function $g_4$ as described in the ``expected proof'' section above, and let $B$ answer according to $g_4$. The checks in \lref[Item]{itm:sumcheck-step} will always succeed. It remains to take $A$ to be the restrictions of $B$ to the manifolds and then the consistency checks will pass and the verifier will always output as required.
\end{proof}

\begin{lemma}[Distributional Soundness]\label{lem:SC-soundness}
The verifier above has soundness error at most $\delta = \card \F^{-0.1}$. Namely, given $(\Phi,F)$ for every proof $\Pi=(A,B)$, there are functions $\fakePi(\cdot),\fakex(\cdot)$ such that
\phnote{Is $\Phi(x) = 0$ or 1?}
  \begin{itemize}
  \item For each $R$, either $\Phi(\fakex(R))=1$
  and $\fakePi(R)$ is a valid proof for ``$x\in SAT(\Phi)$'' or $\fakePi(R)=\bot$.
  \item For every $j$, there is probability at least $1-\epsilon$ that when
  $R$ is chosen randomly and $V$ is run on $(\Phi,F;j,R)$ it
  either rejects, or $\fakePi(R)$ is a proof that completely agrees with
  the answers of the provers $A,B$ on the queries of $V$ (in
  which case $V$'s output is consistent with $\fakex(R)$).
\end{itemize}
\end{lemma}

\begin{proof}
Fix $\Pi = (A,B)$. Given $B$, let $Q_1,\ldots,Q_L$ be degree $\le d$ functions as in \lref[Lemma]{lem:LDT}. We say that $Q_i$ is a valid proof for $a\in \Phi^{-1}(1)$ when the $B$ prover answers according to $Q_i$, there is an $A$ prover causing the verifier to always output consistently with $LDE_a$. Let $I\subset [L]$ be the indices for which $Q_i$ is a valid proof for some $a,b$.

For each $R$ note that in the verifier protocol $x_1$ is chosen (based on $R$ but) independently of $j$. Set $\fakePi(R)=\bot$ if events $E1$ or $E2$ occurred, where
\begin{enumerate}
  \item[E1:] $B(x_1) \not \in \set{Q_i(x_1)\;|\;i\in I}$.
  \item[E2:] there is more than one index $i\in I$ for which $B(x_1)=Q_i(x_1)$.
\end{enumerate}
Otherwise, there is a unique $i\in I$ such that $B(x_1)=Q_i(x_1)$. By assumption $Q_i$ is a valid proof for some $a\in\Phi^{-1}(1)$ so we set $\fakex(R)=a$ and set $\fakePi(R) = (A_R,B_R)$ to be a valid proof for $a$.

Now fix an arbitrary $j\in \F^{m}$, and let $R$ be chosen uniformly at random. We claim that the probability that the verifier accepts and yet the view of $\Pi$ and of $\fakePi(R)$ differ is very small. We analyze two cases.
\begin{itemize}
  \item {\bf Accept and $\fakePi(R)=\bot$:} This event can be bounded by
  \[ \Pr[ \hbox{Accept and }E1] +\Pr [ E2 ]  \le \max\left( O(\card\F)^{-0.1})  , O({Lmd}/{\card\F }) \right)+ \binom{L}2\cdot d /\card\F \]
  where the second item is bounded due to the large distance between degree $d$ functions, and the first item is bounded as follows. If $B(x_1) \not\in \set{Q_i(x_1)\;|\; i\in [L]}$ then \lref[Lemma]{lem:LDT} with parameters $m=m_4,k'=k+\ell= O(hm),d$ implies that the probability of acceptance is at most
  \[(mk'd/\card\F)^{1/8} = O(h^2 m^3 \card\F)^{1/8}\le \card \F^{-0.1}\]
(the last inequality is true since $h \le \card\F^{0.01}$ and for large enough $n$ since $m = \log n/\log h$ and $\card \F \gg \poly\log n$.)

  If however $B(x_1) = Q_i(x_1)$ for some $i\in [L]\setminus I$ then for each $i\in [L]\setminus I$ \lref[Lemma]{lem:sumcheck} below shows that the acceptance probability is at most $O(md/\card\F)$, and we take a union bound over all such $i$.
  \item {\bf Accept and $\fakePi(R)|_q \neq \Pi|_q$}: We defined $\fakePi(R)$ so that $B_R(x_1) = B(x_1) =Q_i(x_1)$ for some $i\in I$. So this event occurs if $\ag \neq Q_i|_\gamma$. We observe that this event is contained in $\cup_{i\in I} E_i$ where $E_i$ is the event that $\ag \neq Q_i|_\gamma$ yet $\ag(x_1) = Q_i(x_1)$. For each $i$ this event has probability at most $dd'/\card\F$, and we take a union bound over $i\in I$. The total probability of error in this event is at most $Ldd'/\card\F$.
\end{itemize}
\end{proof}

%
%
\begin{lemma}[Soundness against a low degree prover]\label{lem:sumcheck}
Suppose that $B:\F^{m_4}\to\F$ is a function of degree at most $d = ht+2h\m$, and let $g,\set{s_i^p}$ be its unbundling. Suppose further that $g$ is consistent with $a,b$ such that either $\Phi(a)=0$ or $b\neq F(a)$. Then, for all provers $A$, the probability that the verifier accepts is at most $ O(md/{\card \F})$.
\end{lemma}
%
%
\begin{proof}(of \lref[Lemma]{lem:sumcheck})
Assume that $\Phi(a)=0$ or $F(a)\neq b$ and denote $\sigma=abs$. The probability that a random quadratic $p$ drawn according to $V_0$ (from \lref[Lemma]{lem:had-easy-ver}) will satisfy $p(\sigma)=0$ is at most $O(1/\card\F)$. Suppose $p(\sigma)\neq 0$. This means that
\begin{equation}\label{eq:quadzero}
\sum_{x\in H^\m } \hat p(x) g(x) \neq 0 .
\end{equation}
Observe that if the check in \lref[Item]{itm:4c} passes then either
\begin{equation}\label{eq:sc-m}
s_\m^p \,\neq\, \hat p \cdot g,
\end{equation}
or, for some $i$, as functions of $x_1,\ldots,x_{i-1}$,
\begin{equation}\label{eq:sc-i}
s^p_{i-1}(x_1,\ldots,x_{i-1},\bar  0) \neq \sum_{h\in H} s_i^p(x_1,\ldots,x_{i-1},h,\bar 0).
\end{equation}
Otherwise,
\[ \sum_{x\in H^\m } \hat p(x) g(x) =  \sum_{x_1,\ldots,x_\m\in H } s_\m^p(x_1,\ldots,x_\m) = \]
\[ = \sum_{x_1,\ldots,x_{\m-1} \in H } s_{\m-1}^p(x_1,\ldots,x_{\m-1},0) = ... = \sum_{x_1\in H } s_{\m-1}^p(x_1,\bar 0) = 0
\]
contradicting \eqref{eq:quadzero}.
The verifier checks each of these $\m$ equalities in \eqref{eq:sc-m} and \eqref{eq:sc-i}) on a random point (in Items~\ref{itm:4a} and \ref{itm:4b}), so the probability of acceptance is at most $\m \cdot \frac d{\card \F}$.
\end{proof}
%
%
%
%



{\small
\bibliographystyle{prahladhurl}
\bibliography{DHK}

\newcommand{\etalchar}[1]{$^{#1}$}
\begin{thebibliography}{BGKW88}

\bibitem[ALM{\etalchar{+}}98]{AroraLMSS1998}
\textsc{Sanjeev Arora}, \textsc{Carsten Lund}, \textsc{Rajeev Motwani},
  \textsc{Madhu Sudan}, and \textsc{Mario Szegedy}.
\newblock \emph{Proof verification and the hardness of approximation problems}.
\newblock J. ACM, 45(3):501--555, May 1998.
\newblock (Preliminary version in {\em 33rd FOCS}, 1992).
\newblock \href{http://eccc.hpi-web.de/report/1998/008}{\path{eccc:TR98-008}},
  \href{http://dx.doi.org/10.1145/278298.278306}{\path{doi:10.1145/278298.278306}}.

\bibitem[AS98]{AroraS1998}
\textsc{Sanjeev Arora} and \textsc{Shmuel Safra}.
\newblock \emph{Probabilistic checking of proofs: A new characterization
  of~{NP}}.
\newblock J. ACM, 45(1):70--122, January 1998.
\newblock (Preliminary version in {\em 33rd FOCS}, 1992).
\newblock
  \href{http://dx.doi.org/10.1145/273865.273901}{\path{doi:10.1145/273865.273901}}.

\bibitem[AS03]{AroraS2003}
\textsc{Sanjeev Arora} and \textsc{Madhu Sudan}.
\newblock \emph{Improved low-degree testing and its applications}.
\newblock Combinatorica, 23(3):365--426, 2003.
\newblock (Preliminary version in {\em 29th STOC}, 1997).
\newblock \href{http://eccc.hpi-web.de/report/1997/003}{\path{eccc:TR97-003}},
  \href{http://dx.doi.org/10.1007/s00493-003-0025-0}{\path{doi:10.1007/s00493-003-0025-0}}.

\bibitem[BGH{\etalchar{+}}06]{BenSassonGHSV2006}
\textsc{Eli {Ben-Sasson}}, \textsc{Oded Goldreich}, \textsc{Prahladh Harsha},
  \textsc{Madhu Sudan}, and \textsc{Salil Vadhan}.
\newblock \emph{Robust {PCP}s of proximity, shorter {PCP}s and applications to
  coding}.
\newblock SIAM J. Comput., 36(4):889--974, 2006.
\newblock (Preliminary version in {\em 36th STOC}, 2004).
\newblock \href{http://eccc.hpi-web.de/report/2004/021}{\path{eccc:TR04-021}},
  \href{http://dx.doi.org/10.1137/S0097539705446810}{\path{doi:10.1137/S0097539705446810}}.

\bibitem[BGKW88]{BenOrGKW1988}
\textsc{Michael {Ben-Or}}, \textsc{Shafi Goldwasser}, \textsc{Joe Kilian}, and
  \textsc{Avi Wigderson}.
\newblock \emph{Multi-prover interactive proofs: How to remove intractability
  assumptions}.
\newblock In \emph{Proc.\ $20$th ACM Symp.\ on Theory of Computing (STOC)},
  pages 113--131. 1988.
\newblock
  \href{http://dx.doi.org/10.1145/62212.62223}{\path{doi:10.1145/62212.62223}}.

\bibitem[BGLR93]{BellareGLR1993}
\textsc{Mihir Bellare}, \textsc{Shafi Goldwasser}, \textsc{Carsten Lund}, and
  \textsc{Alexander Russell}.
\newblock \emph{Efficient probabilistically checkable proofs and applications
  to approximation}.
\newblock In \emph{Proc.\ $25$th ACM Symp.\ on Theory of Computing (STOC)},
  pages 294--304. 1993.
\newblock
  \href{http://dx.doi.org/10.1145/167088.167174}{\path{doi:10.1145/167088.167174}}.

\bibitem[BS08]{BenSassonS2008}
\textsc{Eli {Ben-Sasson}} and \textsc{Madhu Sudan}.
\newblock \emph{Short {PCP}s with polylog query complexity}.
\newblock SIAM J. Comput., 38(2):551--607, 2008.
\newblock (Preliminary version in {\em 37th STOC}, 2005).
\newblock \href{http://eccc.hpi-web.de/report/2004/060}{\path{eccc:TR04-060}},
  \href{http://dx.doi.org/10.1137/050646445}{\path{doi:10.1137/050646445}}.

\bibitem[CK09]{ChuzhoyK2009}
\textsc{Julia Chuzhoy} and \textsc{Sanjeev Khanna}.
\newblock \emph{Polynomial flow-cut gaps and hardness of directed cut
  problems}.
\newblock J. ACM, 56(2), 2009.
\newblock (Preliminary version in {\em 39th STOC}, 2007).
\newblock
  \href{http://dx.doi.org/10.1145/1502793.1502795}{\path{doi:10.1145/1502793.1502795}}.

\bibitem[DFK{\etalchar{+}}11]{DinurFKRS2011}
\textsc{Irit Dinur}, \textsc{Eldar Fischer}, \textsc{Guy Kindler}, \textsc{Ran
  Raz}, and \textsc{Shmuel Safra}.
\newblock \emph{{PCP} characterizations of {NP}: Toward a polynomially-small
  error-probability}.
\newblock Comput.\ Complexity, 20(3):413--504, 2011.
\newblock (Preliminary version in {\em 31st STOC}, 1999).
\newblock \href{http://eccc.hpi-web.de/report/1998/066}{\path{eccc:TR98-066}},
  \href{http://dx.doi.org/10.1007/s00037-011-0014-4}{\path{doi:10.1007/s00037-011-0014-4}}.

\bibitem[DH13]{DinurH2013}
\textsc{Irit Dinur} and \textsc{Prahladh Harsha}.
\newblock \emph{Composition of low-error 2-query {PCP}s using decodable
  {PCP}s}.
\newblock SIAM J. Comput., 42(6):2452–--2486, 2013.
\newblock (Preliminary version in {\em 51st FOCS}, 2009).
\newblock \href{http://eccc.hpi-web.de/report/2009/042}{\path{eccc:TR09-042}},
  \href{http://dx.doi.org/10.1137/100788161}{\path{doi:10.1137/100788161}}.

\bibitem[DHK15]{DinurHK2015}
\textsc{Irit Dinur}, \textsc{Prahladh Harsha}, and \textsc{Guy Kindler}.
\newblock \emph{Polynomially low error {PCP}s with polyloglog n queries via
  modular composition}.
\newblock In \emph{Proc.\ $47$th ACM Symp.\ on Theory of Computing (STOC)}.
  2015.
\newblock (To appear).
\newblock
  \href{http://dx.doi.org/10.1145/2746539.2746630}{\path{doi:10.1145/2746539.2746630}}.

\bibitem[Din07]{Dinur2007}
\textsc{Irit Dinur}.
\newblock \emph{The {PCP} theorem by gap amplification}.
\newblock J. ACM, 54(3):12, 2007.
\newblock (Preliminary version in {\em 38th STOC}, 2006).
\newblock \href{http://eccc.hpi-web.de/report/2005/046/}{\path{eccc:TR05-046}},
  \href{http://dx.doi.org/10.1145/1236457.1236459}{\path{doi:10.1145/1236457.1236459}}.

\bibitem[DM11]{DinurM2011}
\textsc{Irit Dinur} and \textsc{Or~Meir}.
\newblock \emph{Derandomized parallel repetition via structured {PCP}s}.
\newblock Comput.\ Complexity, 20(2):207--327, 2011.
\newblock (Preliminary version in {\em 25th Conference on Computation
  Complexity}, 2010).
\newblock \href{http://arxiv.org/abs/1002.1606}{\path{arXiv:1002.1606}},
  \href{http://dx.doi.org/10.1007/s00037-011-0013-5}{\path{doi:10.1007/s00037-011-0013-5}}.

\bibitem[DR06]{DinurR2006}
\textsc{Irit Dinur} and \textsc{Omer Reingold}.
\newblock \emph{Assignment testers: Towards a combinatorial proof of the {PCP}
  {T}heorem}.
\newblock SIAM J. Comput., 36:975--1024, 2006.
\newblock (Preliminary version in {\em 45th FOCS}, 2004).
\newblock
  \href{http://dx.doi.org/10.1137/S0097539705446962}{\path{doi:10.1137/S0097539705446962}}.

\bibitem[DS04]{DinurS2004}
\textsc{Irit Dinur} and \textsc{Shmuel Safra}.
\newblock \emph{On the hardness of approximating label-cover}.
\newblock Inform.\ Process.\ Lett., 89(5):247--254, March 2004.
\newblock \href{http://eccc.hpi-web.de/report/1999/015}{\path{eccc:TR99-015}},
  \href{http://dx.doi.org/10.1016/j.ipl.2003.11.007}{\path{doi:10.1016/j.ipl.2003.11.007}}.

\bibitem[FGL{\etalchar{+}}96]{FeigeGLSS1996}
\textsc{Uriel Feige}, \textsc{Shafi Goldwasser}, \textsc{L{\'a}szl{\'o}
  Lov{\'a}sz}, \textsc{Shmuel Safra}, and \textsc{Mario Szegedy}.
\newblock \emph{Interactive proofs and the hardness of approximating cliques}.
\newblock J. ACM, 43(2):268--292, March 1996.
\newblock (Preliminary version in {\em 32nd FOCS}, 1991).
\newblock
  \href{http://dx.doi.org/10.1145/226643.226652}{\path{doi:10.1145/226643.226652}}.

\bibitem[FK95]{FeigeK1995}
\textsc{Uriel Feige} and \textsc{Joe Kilian}.
\newblock \emph{Impossibility results for recycling random bits in two-prover
  proof systems}.
\newblock In \emph{Proc.\ $27$th ACM Symp.\ on Theory of Computing (STOC)},
  pages 457--468. 1995.
\newblock
  \href{http://dx.doi.org/10.1145/225058.225183}{\path{doi:10.1145/225058.225183}}.

\bibitem[Har10]{Harsha2010}
\textsc{Prahladh Harsha}.
\newblock \href{http://www.tcs.tifr.res.in/~prahladh/teaching/2009-10/limits/}
  {\emph{Limits of approximation algorithsm: {PCP}s and unique games.}}, 2010.
\newblock A course on PCPs at {TIFR} and {IMSc}.

\bibitem[Mos14]{Moshkovitz2014}
\textsc{Dana Moshkovitz}.
\newblock \emph{An approach to the {S}liding {S}cale {C}onjecture via parallel
  repetition for low degree testing}.
\newblock Technical Report TR14-030, Elect.\ Colloq.\ on Comput.\ Complexity
  (ECCC), 2014.
\newblock \href{http://eccc.hpi-web.de/report/2014/030}{\path{eccc:TR14-030}}.

\bibitem[MR10]{MoshkovitzR2010b}
\textsc{Dana Moshkovitz} and \textsc{Ran Raz}.
\newblock \emph{Two-query {PCP} with subconstant error}.
\newblock J. ACM, 57(5), 2010.
\newblock (Preliminary version in {\em 49th FOCS}, 2008).
\newblock \href{http://eccc.hpi-web.de/report/2008/071}{\path{eccc:TR08-071}},
  \href{http://dx.doi.org/10.1145/1754399.1754402}{\path{doi:10.1145/1754399.1754402}}.

\bibitem[Raz98]{Raz1998}
\textsc{Ran Raz}.
\newblock \emph{A parallel repetition theorem}.
\newblock SIAM J. Comput., 27(3):763--803, June 1998.
\newblock (Preliminary version in {\em 27th STOC}, 1995).
\newblock
  \href{http://dx.doi.org/10.1137/S0097539795280895}{\path{doi:10.1137/S0097539795280895}}.

\bibitem[RS97]{RazS1997}
\textsc{Ran Raz} and \textsc{Shmuel Safra}.
\newblock \emph{A sub-constant error-probability low-degree test, and a
  sub-constant error-probability {PCP} characterization of~{NP}}.
\newblock In \emph{Proc.\ $29$th ACM Symp.\ on Theory of Computing (STOC)},
  pages 475--484. 1997.
\newblock
  \href{http://dx.doi.org/10.1145/258533.258641}{\path{doi:10.1145/258533.258641}}.

\bibitem[Sze99]{Szegedy1999}
\textsc{Mario Szegedy}.
\newblock \emph{Many-valued logics and holographic proofs}.
\newblock In \textsc{Jir\'{\i} Wiedermann}, \textsc{Peter van Emde~Boas}, and
  \textsc{Mogens Nielsen}, eds., \emph{Proc.\ $26$th International Colloq.\ of
  Automata, Languages and Programming (ICALP)}, volume 1644 of \emph{LNCS},
  pages 676--686. Springer, 1999.
\newblock
  \href{http://dx.doi.org/10.1007/3-540-48523-6_64}{\path{doi:10.1007/3-540-48523-6_64}}.

\end{thebibliography}

}

\end{document}